\documentclass[a4paper]{llncs} 
\pdfoutput=1
\usepackage{breakcites}
\usepackage[activate={true,nocompatibility},final,tracking=true,kerning=true,spacing=true,factor=1100,stretch=10,shrink=10]{microtype}

\usepackage{etoolbox}
\newbool{anonymous}
\setbool{anonymous}{false}

\usepackage{amsmath}
\usepackage{amssymb}

\let\proof\relax
\let\endproof\relax
\usepackage{amsthm}
\usepackage{thmtools}
\usepackage{thm-restate}

\usepackage{mathtools}
\usepackage{mathrsfs}
\usepackage{stmaryrd}
\usepackage{xifthen}
\usepackage{xspace}
\newenvironment{proof-sketch}{\proof}{\endproof}

\let\originalleft\left
\let\originalright\right
\renewcommand{\left}{\mathopen{}\mathclose\bgroup\originalleft}
\renewcommand{\right}{\aftergroup\egroup\originalright}

\usepackage{tikz}

\newcommand{\gh}{\mathit{GH}}                

\newcommand{\cnot}{\mathsf{CNOT}}
\newcommand{\X}{\mathsf{X}}
\newcommand{\Y}{\mathsf{Y}}
\newcommand{\Z}{\mathsf{Z}}
\renewcommand{\H}{\mathsf{H}}
\renewcommand{\P}{\mathsf{P}}
\newcommand{\T}{{\sf T}\xspace}
\newcommand{\I}{\mathsf{I}}
\newcommand{\C}{\mathsf{C}}
\newcommand{\U}{\mathsf{U}}

\newcommand{\OR}{\mathsf{OR}}

\newcommand{\Eval}{{\mathsf{Eval}}}
\newcommand{\KeyGen}{{\mathsf{KeyGen}}}
\newcommand{\CL}{{\sf CL}\xspace}
\newcommand{\QHE}{{\sf QHE}\xspace}
\newcommand{\HE}{{\sf HE}}
\newcommand{\AUX}{{\sf AUX}\xspace}

\newcommand{\EPR}{{\sf EPR}\xspace}
\newcommand{\Enc}{{\sf Enc}}
\newcommand{\Dec}{{\sf Dec}}
\newcommand{\Rec}{{\sf Rec}}
\newcommand{\sk}{\mathit{sk}}
\newcommand{\pk}{\mathit{pk}}
\newcommand{\evk}{\mathit{evk}}
\newcommand{\id}{{\mathbb{I}}}

\newcommand{\TOY}{{\sf TOY}}

\newcommand{\NC}{{\sf NC}}

\newcommand{\meas}{
\begin{tikzpicture}
\filldraw[fill=white] (0,.25) rectangle (.7,-.25);
\draw (.67,-.1) arc (50:130:.5);
\draw (.35,-.2)--(.525,.2);
\end{tikzpicture}
}

\newcommand{\advA}{{\mathscr{A}}}

\newcommand{\bra}[1]{{\left\langle{#1}\right\vert}}
\newcommand{\ket}[1]{{\left\vert{#1}\right\rangle}}
\newcommand{\proj}[1]{\ket{#1}\bra{#1}}

\newcommand{\SCHEME}{{\sf TP}\xspace}

\newcommand{\GenGadget}{{\mathsf{GenGadget}}}
\newcommand{\GenMeasurement}{{\mathsf{GenMeasurement}}}

\newcommand{\Sim}{\mathsf{Sim}}

\newcommand{\encrypted}[2][]{
\ifthenelse{\equal{#1}{}}
     {\widetilde{#2}}
     {\widetilde{#2}^{[#1]}}
}
\newcommand{\PubK}[1]{\mathsf{PubK}^{\mathsf{cpa}}_{#1}(\kappa)}
\newcommand{\PubKm}[1]{\mathsf{PubK}^{\mathsf{cpa-mult}}_{#1}(\kappa)}

\usetikzlibrary{decorations.pathmorphing}

\newenvironment{outdent}
               {\list{}{\leftmargin-2cm
                \rightmargin\leftmargin}                \item\relax}
               {\endlist}

\usepackage{hyperref}
\hypersetup{
    colorlinks=true,
    linkcolor=black,
    citecolor=black,
    filecolor=black,
    urlcolor=black,
}               

\begin{document}
\title{Quantum homomorphic encryption for polynomial-sized circuits}
\ifbool{anonymous}{\author{ } \institute{ }}{
\author{Yfke Dulek\inst{1,3} \and Christian Schaffner\inst{1,2,3} \and Florian Speelman\inst{2,3}}
\institute{
University of Amsterdam
\and
CWI, Amsterdam
\and
\href{http://www.qusoft.org/}{QuSoft}
}}
\date{\today}
\maketitle

\thispagestyle{plain} 
\begin{abstract}
We present a new scheme for quantum homomorphic encryption which is compact and allows for efficient evaluation of arbitrary polynomial-sized quantum circuits. 
Building on the framework of Broadbent and Jeffery~\cite{BJ15} and recent results in the area of instantaneous non-local quantum computation~\cite{Spe15arxiv}, we show how to construct quantum gadgets that allow perfect correction of the errors which occur during the homomorphic evaluation of \T gates on encrypted quantum data.
Our scheme can be based on any classical (leveled) fully homomorphic encryption (FHE) scheme and requires no computational assumptions besides those already used by the classical scheme. 
The size of our quantum gadget depends on the space complexity of the classical decryption function -- which aligns well with the current efforts to minimize the complexity of the decryption function. 

Our scheme (or slight variants of it) offers a number of additional advantages such as ideal compactness,
the ability to supply gadgets ``on demand'', circuit privacy for the evaluator against passive adversaries,
and a three-round scheme for blind delegated quantum computation which puts only very limited demands on the quantum abilities of the client.

\end{abstract}
{\bf Keywords:} Homomorphic encryption, quantum cryptography, quantum teleportation, garden-hose model 

\newpage
\section{Introduction}
Fully homomorphic encryption (FHE) is the holy grail of modern cryptography. Rivest, Adleman and Dertouzous were the first to observe the possibility of manipulating encrypted data in a meaningful way, rather than just storing and retrieving it~\cite{RAD78}. After some partial progress~\cite{GM84,Pai99,BGN05,IP07} over the years, a breakthrough happened in 2009 when Gentry presented a fully-homomorphic encryption (FHE) scheme~\cite{Gen09}. Since then, FHE schemes have been simplified~\cite{DGHV10} and based on more standard assumptions~\cite{BV11}. The exciting developments around FHE have sparked a large amount of research in other areas such as functional encryption~\cite{GKPVZ13b,GVW13,GKPVZ13,SW14} and obfuscation~\cite{GGHRSW13}. 

Developing quantum computers is a formidable technical challenge, so it currently seems likely that quantum computing will not be available immediately to everyone and hence quantum computations have to be outsourced. Given the importance of classical\footnote{Here and throughout the article, we use ``classical'' to mean ``non-quantum''.} FHE for ``computing in the cloud'', it is natural to wonder about the existence of encryption schemes which can encrypt \emph{quantum data} in such a way that a server can carry out arbitrary \emph{quantum computations} on the encrypted data (without interacting with the encrypting party\footnote{In contrast to \emph{blind} or \emph{delegated quantum computation} where some interaction between client and server is usually required, see Section~\ref{sec:relatedwork} for references.}). While previous work on \emph{quantum homomorphic encryption} has mostly focused on information-theoretic security (see Section~\ref{sec:relatedwork} below for details), schemes that are based on computational assumptions have only recently been thoroughly investigated by Broadbent and Jeffery. In~\cite{BJ15}, they give formal definitions of quantum fully homomorphic encryption (QFHE) and its security and they propose three schemes for quantum homomorphic encryption assuming the existence of classical FHE.

A natural idea is to encrypt a message qubit with the quantum one-time pad (i.e.\ by applying a random Pauli operation), and send the classical keys for the quantum one-time pad along as classical information, encrypted by the classical FHE scheme. This basic scheme is called \CL in~\cite{BJ15}. It is easy to see that \CL allows an evaluator to compute arbitrary Clifford operations on encrypted qubits, simply by performing the actual Clifford circuit, followed by homomorphically updating the quantum one-time pad keys according to the commutation rules between the performed Clifford gates and the Pauli encryptions. The \CL scheme can be regarded as analogous to additively homomorphic encryption schemes in the classical setting. The challenge, like multiplication in the classical case, is to perform non-Clifford gates such as the \T gate. Broadbent and Jeffery propose two different approaches for doing so, accomplishing homomorphic encryption for circuits with a limited number of \T gates. These results lead to the following main open problem:

\begin{center}
{\it Is it possible to construct a quantum homomorphic scheme that allows evaluation of polynomial-sized quantum circuits?}
\end{center}

\subsection{Our Contributions}\label{sec:contributions}
We answer the above question in the affirmative by presenting a new scheme \SCHEME (as abbreviation for teleportation) for quantum homomorphic encryption which is both compact and efficient for circuits with polynomially many \T gates. The scheme is secure against chosen plaintext attacks from quantum adversaries, as formalized by the security notion \emph{q-IND-CPA security} defined by Broadbent and Jeffery \cite{BJ15}.

Like the schemes proposed in \cite{BJ15}, our scheme is an extension of the Clifford scheme \CL. We add auxiliary quantum states to the evaluation key which we call quantum \emph{gadgets} and which aid in the evaluation of the \T gates. The size of a gadget depends only on (a certain form of) the space complexity of the decryption function of the classical FHE scheme. This relation turns out to be very convenient, as classical FHE schemes are often optimized with respect to the complexity of the decryption operation (in order to make them bootstrappable). As a concrete example, if we instantiate our scheme with the classical FHE scheme by Brakerski and Vaikuntanathan~\cite{BV11}, each evaluation gadget of our scheme consists of a number of qubits which is polynomial in the security parameter.

In \SCHEME, we require exactly one evaluation gadget for every \T gate that we would like to evaluate homomorphically. Intuitively, after a \T gate is performed on a one-time-pad encrypted qubit $\X^a\Z^b\ket{\psi}$, the result might contain an unwanted phase $\P^a$ depending on the key $a$ with which the qubit was encrypted, since $\T \X^a \Z^b \ket{\psi} = \P^a \X^a \Z^{b} \T  \ket{\psi}$. Obviously, the evaluator is not allowed to know the key $a$. Instead, he holds an encryption $\tilde{a}$ of the key, produced by a classical FHE scheme. The evaluator can teleport the encrypted qubit ``through the gadget"~\cite{GC99} in a way that depends on $\tilde{a}$, in order to remove the unwanted phase. In more detail, the quantum part of the gadget consists of a number of EPR pairs which are prepared in a way that depends on the secret key of the classical FHE scheme. Some classical information is provided with the gadget that allows the evaluator to homomorphically update the encryption keys after the teleportation steps.
On a high level, the use of an evaluation gadget corresponds to a \emph{instantaneous non-local quantum computation}\footnote{This
term is not related to the term `instantaneous quantum computation' \cite{SB08}, and instead first
was used as a specific form of non-local quantum computation, one where all parties have to act simultaneously.} where one party holds the secret key of the classical FHE scheme, and the other party holds the input qubit and a classical encryption of the key to the quantum one-time pad. Together, this information determines whether an inverse phase gate $\P^{\dag}$ needs to be  performed on the qubit or not.
Very recent results by Speelman~\cite{Spe15arxiv} show how to perform such computations with a bounded amount of entanglement. These techniques are the crucial ingredients for our construction and are the reason why the \emph{garden-hose complexity}~\cite{BFSS13} of the decryption procedure of the classical FHE is related to the size of our gadgets.

The quantum part of our evaluation gadget is strikingly simple, which provides a number of advantages. To start with, the evaluation of a \T gate requires only one gadget, and does not cause errors to accumulate on the quantum state. The scheme is very compact in the sense that the state of the system after the evaluation of a \T gate has the same form as after the initial encryption, except for any classical changes caused by the classical FHE evaluation. This kind of compactness also implies that individual evaluation gadgets can be supplied ``on demand'' by the holder of the secret key. Once an evaluator runs out of gadgets, the secret key holder can simply supply more of them.

Furthermore, $\SCHEME$ does not depend on a specific classical FHE scheme, hence any advances in classical FHE can directly improve our scheme.
Our requirements for the classical FHE scheme are quite modest: we only require the classical scheme to have a space-efficient decryption procedure and to be secure against quantum adversaries.
In particular, no circular-security assumption is required. Since we supply at most a polynomial number of evaluation gadgets, our scheme $\SCHEME$ is leveled homomorphic by construction, and we can simply switch to a new classical key after every evaluation gadget.
In fact, the Clifford gates in the quantum evaluation circuit only require additive operations from the classical homomorphic scheme, while each $\T$ gate needs a fixed (polynomial) number of multiplications.
Hence, we do not actually require fully homomorphic classical encryption, but leveled fully homomorphic schemes suffice.

Finally,     circuit privacy in the passive setting almost comes for free. When wanting to hide which circuit was evaluated on the data, the evaluating party can add an extra randomization layer to the output state by applying his own one-time pad. We show that if the classical FHE scheme has the circuit-privacy property, then this extra randomization completely hides the circuit from the decrypting party. This is not unique to our specific scheme: the same is true for \CL.

\medskip

In terms of applications, our construction can be appreciated as a constant-round scheme for \emph{blind delegated quantum computation}, using computational assumptions. The server can evaluate a universal quantum circuit on the encrypted input, consisting of the client's quantum input and a (classical) description of the client's circuit. In this context, it is desirable to minimize the quantum resources needed by the client. We argue that our scheme can still be used for constant-round blind delegated quantum computation if we limit either the client's quantum memory or the types of quantum operations the client can perform.

As another application, we can instantiate our construction with a classical FHE scheme that allows for \emph{distributed} key generation and decryption amongst different parties that all hold a share of the secret key~\cite{AJLTVW12}. In that case, it is likely that our techniques can be adapted to perform \emph{multiparty quantum computation}~\cite{BCGHS06} in the semi-honest case. However, the focus of this article lies on the description and security proof of the new construction, and more concrete applications are the subject of upcoming work.

\subsection{Related Work} \label{sec:relatedwork}

Early classical FHE schemes were limited in the sense that they could not facilitate arbitrary operations on the encrypted data: some early schemes only implemented a single operation (addition or multiplication)\cite{RSA78,GM84,Pai99}; later on it became possible to combine several operations in a limited way \cite{BGN05,GHV10,SYY99}.  Gentry's first fully homomorphic encryption scheme \cite{Gen09} relied on several non-standard computational assumptions. Subsequent work~\cite{BGV12,BV11} has relaxed these assumptions or replaced them with more conventional assumptions such as the hardness of learning with errors (LWE), which is believed to be hard also for quantum attackers.
 It is impossible to completely get rid of computational assumptions for a classical FHE scheme, since the existence of such a scheme would imply the existence of an information-theoretically secure protocol for private information retrieval (PIR)~\cite{KO97} that breaks the lower bound on the amount of communication required for that task \cite{CKG+98,Fil12}.

While quantum fully homomorphic encryption (QFHE) is closely related to the task of blind or delegated quantum computation \cite{Chi05,BFK09,ABE10,VFPR14,FBS+14,Bro15,Lia15}, QFHE does not allow interaction between the client and the server during the computation. Additionally, in QFHE, the server is allowed to choose which unitary it wants to apply to the (encrypted) data.

Yu, P\'{e}rez-Delgado and Fitzsimons \cite{YPF14} showed that perfectly in\-for\-ma\-tion-theo\-ret\-i\-cally secure QFHE is not possible unless the size of the encryption grows exponentially in the input size. Thus, any scheme that attempts to achieve information-theoretically secure QFHE has to leak some proportion of the input to the server \cite{AS06,RFG12} or can only be used to evaluate a subset of all unitary transformations on the input \cite{RFG12,Lia13,TKO+14}. Like the multiplication operation is hard in the classical case, the hurdle in the quantum case seems to be the evaluation of non-Clifford gates. A recent result by Ouyang, Tan and Fitzsimons provides information-theoretic security for circuits with at most a constant number of non-Clifford operations~\cite{OTF15}.

Broadbent and Jeffery \cite{BJ15} proposed two schemes that achieve homomorphic encryption for nontrivial sets of quantum circuits. Instead of trying to achieve information-theoretic security, they built their schemes based on a classical FHE scheme and hence any computational assumptions on the classical scheme are also required for the quantum schemes. Computational assumptions allow bypassing the impossibility result from \cite{YPF14} and work toward a (quantum) fully homomorphic encryption scheme.

Both of the schemes presented in \cite{BJ15} are extensions of the scheme \CL described in Section~\ref{sec:contributions}. These two schemes use different methods to implement the evaluation of a \T gate, which we briefly discuss here. In the \EPR scheme, some entanglement is accumulated in a special register during every evaluation of a \T gate, and stored there until it can be resolved in the decryption phase. Because of this accumulation, the complexity of the decryption function scales (quadratically) with the number of \T gates in the evaluated circuit, thereby violating the compactness requirement of QFHE. The scheme \AUX also extends \CL, but handles \T gates in a different manner. The evaluator is supplied with auxiliary quantum states, stored in the evaluation key, that allow him to evaluate \T gates and immediately remove any error that may have occurred. In this way, the decryption procedure remains very efficient and the scheme is compact. Unfortunately, the required auxiliary states grow doubly exponentially in size with respect to the \T depth of the circuit, rendering \AUX useful only for circuits with constant \T depth.
Our scheme \SCHEME is related to \AUX in that extra resources for removing errors are stored in the evaluation key. In sharp contrast to \AUX, the size of the evaluation key in \SCHEME only grows linearly in the number of \T gates in the circuit (and polynomially in the security parameter), allowing the scheme to be leveled fully homomorphic. Since the evaluation of the other gates causes no errors on the quantum state, no gadgets are needed for those; any circuit containing polynomially many \T gates can be efficiently evaluated.

\subsection{Structure of this paper}
We start by introducing some notation in Section~\ref{sec:preliminaries} and presenting the necessary preliminaries on quantum computation, (classical and quantum) homomorphic encryption, and the garden-hose model which is essential to the most-general construction of the gadgets. In Section~\ref{sec:scheme}, we describe the scheme \SCHEME and show that it is compact. The security proof of \SCHEME is somewhat more involved, and is presented in several steps in Section~\ref{sec:security}, along with an informal description of a circuit-private variant of the scheme. In Section~\ref{sec:gadget-construction}, the rationale behind the quantum gadgets is explained, and some examples are discussed to clarify the construction. We conclude our work in Section~\ref{sec:conclusion} and propose directions for future research.

\section{Preliminaries}\label{sec:preliminaries}
\subsection{Quantum computation}
We assume the reader is familiar with the standard notions in the field of quantum computation (for an introduction, see~\cite{NC00}). In this subsection, we only mention the concepts that are essential to our construction.

The single-qubit \emph{Pauli group} is, up to global phase, generated by the bit flip and phase flip operations,
\[
\X = \left[
\begin{array}{c c}
0&1\\1&0
\end{array}
\right], \ \ \ 
\Z = \left[
\begin{array}{c c}
1&0\\0&-1
\end{array}
\right].
\]
A \emph{Pauli operator} on $n$ qubits is simply any tensor product of $n$ independent single-qubit Pauli operators. All four single-qubit Pauli operators are of the form $\X^a\Z^b$ with $a,b \in \{0,1\}$. Here, and in the rest of the paper, we ignore the global phase of a quantum state, as it is not observable by measurement.

The \emph{Clifford group} on $n$ qubits consists of all unitaries $\C$ that commute with the Pauli group, i.e.\ the Clifford group is the normalizer of the Pauli group. Since all Pauli operators are of the form $\X^{a_1} \Z^{b_1}\otimes \cdots \otimes \X^{a_n}\Z^{b_n}$, this means that $\C$ is a Clifford operator if for any $a_1,b_1, \dots, a_n, b_n \in \{0,1\}$ there exist $a_1',b_1', \dots, a_n', b_n' \in \{0,1\}$ such that (ignoring global phase):
\[\C(\X^{a_1} \Z^{b_1}\otimes \cdots\otimes \X^{a_n}\Z^{b_n}) = (\X^{a'_1} \Z^{b'_1}\otimes\cdots\otimes\X^{a'_n}\Z^{b'_n})\C .\]
All Pauli operators are easily verified to be elements of the Clifford group, and the entire Clifford group is generated by
\[
\P = \left[
\begin{array}{c c}
1&0\\0&i
\end{array}
\right], \ \ \ 
\H = \frac{1}{\sqrt{2}}\left[
\begin{array}{c c}
1&1\\1&-1
\end{array}
\right], \ \ \  \mbox{and} \ \ 
\cnot = \left[
\begin{array}{c c c c}
1&0&0&0\\
0&1&0&0\\
0&0&0&1\\
0&0&1&0
\end{array}
\right]. \ \ \ 
\]
(See for example~\cite{Got98}.) The Clifford group does not suffice to simulate arbitrary quantum circuits, but by adding any single non-Clifford gate, any quantum circuit can be efficiently simulated with only a small error. As in~\cite{BJ15}, we choose this non-Clifford gate to be the $\T$ gate,
\[
\T = \left[
\begin{array}{c c}
1&0\\0&e^{i\pi/4}
\end{array}
\right].
\]
Note that the $\T$ gate, because it is non-Clifford, does not commute with the Pauli group. More specifically, we have $\T\X^a\Z^b = \P^a\X^a\Z^b\T$. It is exactly the formation of this $\P$ gate that has proven to be an obstacle to the design of an efficient quantum homomorphic encryption scheme.

We use $\ket\psi$ or $\ket\varphi$ to denote pure quantum states. Mixed states are denoted with $\rho$ or $\sigma$. Let $\id_d$ denote the identity matrix of dimension $d$: this allows us to write the \emph{completely mixed state} as $\id_d / d$.

Define $\ket{\Phi^+} := \frac{1}{\sqrt{2}}(\ket{00}+\ket{11})$ to be an EPR pair.

If $X$ is a random variable ranging over the possible basis states $B$ for a quantum system, then let $\rho(X)$ be the density matrix corresponding to $X$, i.e.~$\rho(X) := \sum_{b \in B} \Pr[X = b]\ket b \bra b$.

Applying a Pauli operator that is chosen uniformly at random results in a single-qubit completely mixed state, since
\[
\forall \rho \; : \; \sum_{a,b \in \{0,1\}} \left( \frac{1}{4}  \X^a\Z^b\rho(\X^a\Z^b)^{\dag}\right) = \frac{\id_2}{2}
\]
This property is used in the construction of the \emph{quantum one-time pad}: applying a random Pauli $\X^a\Z^b$ to a qubit completely hides the content of that qubit to anyone who does not know the key $(a,b)$ to the pad. Anyone in possession of the key can decrypt simply by applying $\X^a\Z^b$ again.

\subsection{Homomorphic encryption}\label{sec:he-definition}
This subsection provides the definitions of (classical and quantum) homomorphic encryption schemes, and the security conditions for such schemes. In the current work, we only consider homomorphic encryption in the public-key setting. For a more thorough treatment of these concepts, and how they can be transferred to the symmetric-key setting, see~\cite{BJ15}.

\subsubsection{The classical setting}
A classical homomorphic encryption scheme $\HE$ consists of four algorithms: key generation, encryption, evaluation, and decryption. The key generator produces three keys: a public key and evaluation key, both of which are publicly available to everyone, and a secret key which is only revealed to the decrypting party. Anyone in possession of the public key can encrypt the inputs $x_1,\ldots,x_\ell$, and send the resulting ciphertexts $c_1,\ldots,c_\ell$ to an evaluator who evaluates some circuit $\C$ on them. The evaluator sends the result to a party that possesses the secret key, who should be able to decrypt it to $\C(x_1,\ldots,x_\ell)$.

More formally, $\HE$ consists of the following four algorithms which run in probabilistic polynomial time in terms of their input and parameters~\cite{BV11}:
\begin{description}
\item[$(\pk,\evk,\sk) \leftarrow \HE.\KeyGen(1^{\kappa})$] where $\kappa \in \mathbb{N}$ is the \emph{security parameter}. Three keys are generated: a public key $\pk$, which can be used for the encryption of messages; a secret key $\sk$ used for decryption; and an evaluation key $\evk$ that may aid in evaluating the circuit on the encrypted state. The keys $\pk$ and $\evk$ are announced publicly, while $\sk$ is kept secret.
\item[$c \leftarrow \HE.\Enc_{\pk}(x)$] for some one-bit message $x \in \{0,1\}$. This probabilistic procedure outputs a ciphertext $c$, using the public key $\pk$.
\item[$c' \leftarrow \HE.\Eval^{\C}_\evk(c_1,\ldots,c_\ell)$] uses the evaluation key to output some ciphertext $c'$ which decrypts to the evaluation of circuit $\C$ on the decryptions of $c_1,\ldots,c_\ell$. We will often think of $\Eval$ as an evaluation of a function $f$ instead of some canonical circuit for $f$, and write $\HE.\Eval^{f}_\evk(c_1,\ldots,c_\ell)$ in this case.
\item[$x' \leftarrow \HE.\Dec_{\sk}(c)$] outputs a message $x' \in \{0,1\}$, using the secret key $\sk$.
\end{description}
In principle, $\HE.\Enc_{\pk}$ can only encrypt single bits. When encrypting an $n$-bit message $x \in \{0,1\}^n$, we encrypt the message bit-by-bit, applying the encryption procedure $n$ times. We sometimes abuse the notation $\HE.\Enc_{\pk}(x)$ to denote this bitwise encryption of the string $x$.

For $\HE$ to be a homomorphic encryption scheme, we require \emph{correctness} in the sense that for any circuit $\C$, there exists a negligible\footnote{A \emph{negligible function} $\eta$ is a function such that for every positive integer $d$, $\eta(n) < 1/n^d$ for big enough $n$.} function $\eta$ such that, for any input $x$,
\[
\Pr[\HE.\Dec_{\sk}(\HE.\Eval^{\C}_{\evk}(\HE.\Enc_\pk(x))) \neq \C(x)] \leq \eta(\kappa) \, .
\]
In this article, we assume for clarity of exposition that classical schemes $\HE$ are perfectly correct, and that it is possible to immediately decrypt after encrypting (without doing any evaluation).

Another desirable property is \emph{compactness}, which states that the complexity of the decryption function should not depend on the size of the circuit: a scheme is compact if there exists a polynomial $p(\kappa)$ such that for any circuit $\C$ and any ciphertext $c$, the complexity of applying $\HE.\Dec$ to the result of $\HE.\Eval^{C}(c)$ is at most $p(\kappa)$.

A scheme that is both correct for all circuits and compact, is called \emph{fully} homomorphic. If it is only correct for a subset of all possible circuits (e.g.~all circuits with no multiplication gates) or if it is not compact, it is considered to be a \emph{somewhat} homomorphic scheme. Finally, a \emph{leveled} fully homomorphic scheme is (compact and) homomorphic for all circuits up to a variable depth $L$, which is supplied as an argument to the key generation function~\cite{Vai11}.

We will use the notation $\encrypted{x}$ to denote the result of running $\HE.\Enc_{\pk}(x)$: that is, $\Dec_{\sk}(\encrypted{x}) = x$ with overwhelming probability. In our construction, we will often deal with multiple classical key sets $(\pk_i,\sk_i,\evk_i)_{i \in I}$ indexed by some set~$I$. In that case, we use the notation $\encrypted[i]{x}$ to denote the result of $\HE.\Enc_{\pk_i}(x)$, in order to avoid confusion. Also note that (e.g.) $\pk_i$ does \emph{not} refer to the $i$th bit of the public key: in case we want to refer to the $i$th bit of some string $s$, we will use the notation $s[i]$.

When working with multiple key sets, it will often be necessary to transform an already encrypted message $\encrypted[i]{x}$ into an encryption $\encrypted[j]{x}$ using a different key set $j \neq i$. To achieve this transformation, we define the procedure $\HE.\Rec_{i \to j}$ that can always be used for this \emph{recryption} task as long as we have access to an encrypted version $\encrypted[j]{\sk_{i}}$ of the old secret key $\sk_{i}$. Effectively, $\HE.\Rec_{i \to j}$ homomorphically evaluates the decryption of $\encrypted[i]{x}$:
\[
\HE.\Rec_{i \to j}(\encrypted[i]{x}) := \HE.\Eval^{\HE.\Dec}_{\evk_j}\Bigl(\encrypted[j]{\sk_i}, \HE.\Enc_{\pk_j}(\encrypted[i]{x})\Bigr).
\]

\subsubsection{The quantum setting}
A quantum homomorphic encryption scheme $\QHE$, as defined in~\cite{BJ15}, is a natural extension of the classical case, and differs from it in only a few aspects. The secret and public keys are still classical, but the evaluation key is allowed to be a quantum state. This means that the evaluation key is not necessarily reusable, and can be consumed during the evaluation procedure.
The messages to be encrypted are qubits instead of bits, and the evaluator should be able to evaluate quantum circuits on them.

All definitions given above carry over quite naturally to the quantum setting (see also~\cite{BJ15}):

\begin{description}
\item[$(\pk, \rho_{\evk},\sk) \leftarrow \QHE.\KeyGen(1^{\kappa})$] where $\kappa \in \mathbb{N}$ is the security parameter. In contrast to the classical case, the evaluation key is a quantum state.
\item[$\sigma \leftarrow \QHE.\Enc_{\pk}(\rho)$] produces, for every valid public key $\pk$ and input state $\rho$ from some message space, to a quantum cipherstate $\sigma$ in some cipherspace.
\item[$\sigma' \leftarrow \QHE.\Eval_{\rho_{\evk}}^{\C}(\sigma)$] represents the evaluation of a circuit $\C$. If $\C$ requires $n$ input qubits, then $\sigma$ should be a product of $n$ cipherstates. The evaluation function maps it to a product of $n'$ states in some output space, where $n'$ is the number of qubits that $\C$ would output. The evaluation key $\rho_{\evk}$ is consumed in the process.
\item[$\rho' \leftarrow \QHE.\Dec_{\sk}(\sigma')$] maps a single state $\sigma'$ from the output space to a quantum state $\rho'$ in the message space. Note that if the evaluation procedure $\QHE.\Eval$ outputs a product of $n'$ states, then $\QHE.\Dec$ needs to be run $n'$ times.
\end{description}
The decryption procedure differs from the classical definition in that we require the decryption to happen subsystem-by-subsystem: this is fundamentally different from the more relaxed notion of \emph{indivisible schemes}~\cite{BJ15} where an auxiliary quantum register may be built up for the entire state, and the state can only be decrypted as a whole. In this work, we only consider the divisible definition.

\subsubsection{Quantum security}\label{sec:prelim-security}

The notion of security that we aim for is that of \emph{indistinguishability under chosen-plaintext attacks}, where the attacker may have quantum computational powers (q-IND-CPA). This security notion was introduced in~\cite[Definition 3.3]{BJ15} (see~\cite{GHS15arxiv} for a similar notion of the security of classical schemes against quantum attackers) and ensures semantic security~\cite{ABFGSS16}.
We restate it here for completeness.

\begin{definition}\cite{BJ15}
The \emph{quantum CPA indistinguishability experiment} with respect to a scheme \QHE and a quantum polynomial-time adversary $\advA=(\advA_1, \advA_2)$, denoted by $\PubK{\advA,\mathsf{QHE}}$, is defined by the following procedure:
\begin{enumerate}
\item $\KeyGen(1^\kappa)$ is run to obtain keys $(pk,sk,\rho_{evk})$.
\item Adversary $\advA_1$ is given $(pk,\rho_{evk})$ and outputs a quantum state on $\mathcal{M} \otimes \cal E$.
\item For $r\in \{0,1\}$, let   $\Xi_{\QHE}^{{\sf cpa},r}: D(\mathcal{M}) \rightarrow D(\mathcal{C})$ be:
$\Xi_{\QHE}^{{\sf cpa},0}(\rho)=  \QHE.\Enc_{pk}(\ket{0}\bra{0})$ and $\Xi_{\QHE}^{{\sf cpa},1}(\rho)=  \QHE.\Enc_{pk}(\rho)$.
A random bit $r \in \{0,1\}$ is chosen and $\Xi_{\QHE}^{{\sf cpa},r}$ is applied to the state in  $\mathcal{M}$ (the output being a state in $\mathcal{C}$).
\item Adversary $\advA_2$ obtains the system in $\mathcal{C} \otimes \mathcal{E}$ and outputs a bit $r'$.
\item The output of the experiment is defined to be~1 if~$r'=r$ and~$0$~otherwise.  In case $r=r'$, we say that $\advA$ \emph{wins} the experiment.
\end{enumerate}
\end{definition}

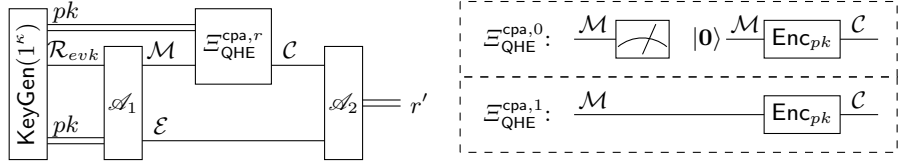
\begin{figure}[h]
\centering
\begin{tikzpicture}
\node at (0,0) {
\begin{tikzpicture}
\draw (-.25,1.54) -- (2,1.54);
\draw (-.25,1.46) -- (2,1.46);
\draw (.5,0)--(3.5,0);
\draw (.5,1)--(3.5,1);
\draw (-.25,0.04)--(.5,0.04);
\draw (-.25,-.04)--(.5,-.04);
\draw (-.25,1) -- (.5,1);
\draw (3.9,.46)--(4.4,.46);
\draw (3.9,.54)--(4.4,.54);

\filldraw[fill=white] (-.75,-.25) rectangle (-.25,1.75);
\node[rotate=90] at (-.5,.75) {\small$\KeyGen(1^\kappa)$};
\node at (0,1.7) {$pk$};
\node at (.1,1.2) {\small $\mathcal{R}_{evk}$};
\node at (0,.2) {$pk$};

\filldraw[fill=white] (.5,1.25) rectangle (1,-.25);
\node at (.75,.5) {$\advA_1$};

\node at (1.25,1.2) {${\small\cal M}$};
\node at (1.25,.2) {${\small\cal E}$};

\filldraw[fill=white] (1.7,1.75) rectangle (2.7,.75);
\node at (2.2,1.25) {$\Xi_{\QHE}^{{\sf cpa},r}$};

\node at (2.95,1.2) {$\small\cal C$};

\filldraw[fill=white] (3.4,1.25) rectangle (3.9,-.25);
\node at (3.65,.5) {$\advA_2$};

\node at (4.65,.5) {$r'$};
\end{tikzpicture}
};

\node at (6,0) {
\begin{tikzpicture}
\node at (0,1) {$\Xi_{\QHE}^{{\sf cpa},0}$:};
\draw (.75,1)--(1.5,1);
\node at (1.65,1) {\meas};
\node at (2.5,1) {$\ket{\mathbf{0}}$};
\draw (2.75,1)--(4.75,1);
\node at (3,1.2) {\small $\cal M$};
\filldraw[fill=white] (3.25,1.25) rectangle (4.25,.75);
\node at (3.75,1) {$\mathsf{Enc}_{pk}$};
\node at (4.5,1.2) {\small $\cal C$};
\node at (1,1.2) {\small $\cal M$};

\node at (0,0) {$\Xi_{\QHE}^{{\sf cpa},1}$:};
\draw (.75,0)--(4.75,0);
\filldraw[fill=white] (3.25,.25) rectangle (4.25,-.25);
\node at (3.75,0) {$\mathsf{Enc}_{pk}$};
\node at (1,.2) {\small $\cal M$};
\node at (4.5,.2) {\small $\cal C$};

\draw[dashed] (-.75,1.5) rectangle (5,.5);

\draw[dashed] (-.75,-.5) rectangle (5,.5);

\end{tikzpicture}
};
\end{tikzpicture}
\caption{\cite[reproduced with permission of the authors]{BJ15} The
  quantum CPA indistinguishability experiment
  $\PubK{\advA,\mathsf{QHE}}$. Double lines represent classical
  information flow, and single lines represent quantum information
  flow. The adversary $\advA$ is split up into two separate algorithms
  $\advA_1$ and $\advA_2$, which share a working memory represented by
  the quantum state in register $\mathcal{E}$.}\label{fig:PubK-def}
\end{figure}

The game  $\PubK{\advA,\mathsf{QHE}}$ is depicted in Figure~\ref{fig:PubK-def}. Informally, the challenger randomly chooses whether to encrypt some message, chosen by the adversary, or instead to encrypt the state $\ket{0}\bra{0}$. The adversary has to guess which of the two happened. If he cannot do so with more than negligible advantage, the encryption procedure is considered to be q-IND-CPA secure:

\begin{definition}\cite[Definition 3.3]{BJ15}
A (classical or quantum) homomorphic encryption scheme $\mathsf{S}$ is \emph{q-IND-CPA} secure if for any quantum polynomial-time adversary $\advA = (\advA_1,\advA_2)$ there exists a negligible function $\eta$ such that:
\[
\Pr[\PubK{\advA, \mathsf{S}} = 1] \leq \frac{1}{2} + \eta(\kappa).
\]
\end{definition}

Analogously to $\PubK{\advA,\mathsf{S}}$, in the game $\PubKm{\advA,\mathsf{S}}$, the adversary can give multiple messages to the challenger, which are either all encrypted, or all replaced by zeros. Broadbent and Jeffery~\cite{BJ15} show that these notions of security are equivalent.

\subsection{Garden-hose complexity}
\label{sec:prelim-gardenhose}

The \emph{garden-hose model} is a model of communication complexity which was introduced by Buhrman, Fehr, Schaffner and Speelman~\cite{BFSS13}
to study a protocol for position-based quantum cryptography. The model recently saw new use, when Speelman~\cite{Spe15arxiv} used it to
construct new protocols for the task of instantaneous non-local quantum computation, thereby breaking a wider class of schemes for position-based quantum cryptography. (Besides the garden-hose model,
this construction used tools from secure delegated computation. These techniques were first used in the setting of instantaneous non-local quantum computation
by Broadbent~\cite{Bro15arxiv}.)

We will not explain the garden-hose model thoroughly, but instead give a short overview. 
The garden-hose model involves two parties, Alice with input $x$ and Bob with input $y$, that jointly want to compute a function $f$.
To do this computation, they are allowed to use garden hoses to link up pipes that run between them, one-to-one,
in a way which depends on their local inputs. Alice also has a water tap, which she connects to one of the pipes. Whenever $f(x,y)=0$, the water has to exit at an open pipe on Alice's side, and whenever $f(x,y)=1$ the water should exit on Bob's side.

The applicability of the garden-hose model to our setting stems from a close correspondence between protocols in the garden-hose model and teleporting
a qubit back-and-forth; the `pipes' correspond to EPR pairs and the `garden hoses' can be translated into Bell measurements. Our construction of the gadgets in Section~\ref{sec:gadget-construction-logspace} will depend on the number of pipes needed to compute the decryption function $\HE.\Dec$ of a classical fully homomorphic encryption scheme. It will turn out that any log-space computable decryption function allows for efficiently constructable polynomial-size gadgets.

\section{The TP scheme}\label{sec:scheme}
Our scheme $\SCHEME$ (for teleportation) is an extension of the scheme $\CL$ presented in \cite{BJ15}: the quantum state is encrypted using a quantum one-time pad, and Clifford gates are evaluated simply by performing the gate on the encrypted state
and then homomorphically updating the encrypted keys to the pad. The new scheme $\SCHEME$, like $\AUX$, includes additional resource states (gadgets) in the evaluation key. These gadgets can be used to immediately correct any $\P$ errors that might be present after the application of a $\T$ gate. The size of the evaluation key thus grows linearly with the upper bound to the number of $\T$ gates in the circuit: for every $\T$ gate the evaluation key contains one gadget, along with some classical information on how to use that gadget.

\subsection{Gadget}\label{sec:gadget}
In this section we only give the general form of the gadget, which suffices to prove security.
The explanation on how to construct these gadgets, which depend on the decryption function of the classical homomorphic scheme $\HE.\Dec$,
is deferred to Section~\ref{sec:gadget-construction}. 

Recall that when a $\T$ gate is applied to the state $\X^a\Z^b\ket\psi$, an unwanted $\P$ error may occur since
$
\T\X^a\Z^b = \P^a\X^a\Z^b\T.
$
If $a$ is known, this error can easily be corrected by applying $\P^{\dag}$ whenever $a = 1$. However, as we will see, the evaluating party only has access to some encrypted version $\encrypted{a}$ of the key $a$,
and hence is not able to decide whether or not to correct the state.

We show how the key generator can create a gadget ahead of time that corrects the state, conditioned on $a$, when the qubit $\P^a\X^a\Z^b\T\ket\psi$ is teleported through it. The gadget will not reveal any information about whether or not a $\P$ gate was present before the correction. Note that the value of $a$ is completely unknown to the key generator, so the gadget cannot depend on it. Instead, the gadget will depend on the secret key $\sk$,
and the evaluator will use it in a way that depends on $\encrypted{a}$.

The intuition behind our construction is as follows. A gadget consists of a set of fully entangled pairs that are crosswise linked
up in a way that depends on the secret key $\sk$ and the decryption function of the classical homomorphic scheme $\HE$. If the decryption function $\HE.\Dec$ is simple enough, i.e.~computable in logarithmic space or by low-depth binary circuit, the size of this state is polynomial in the security parameter.

Some of these entangled pairs have an extra inverse phase gate applied to them. Note that teleporting any qubit $\X^a \Z^b \ket{\psi}$ through, for example, $(\P^{\dagger} \otimes \I) \ket{\Phi^+}$, effectively applies an inverse phase gate to the qubit, which ends up in the state $\X^{a'} \Z^{b'} \P^{\dagger} \ket{\psi}$, where the new Pauli corrections
$a'$,$b'$ depend on $a$,$b$ and the outcome of the Bell measurement.

When wanting to remove an unwanted phase gate, the evaluator of the circuit teleports a qubit through this gadget state in a way which is specified by $\encrypted{a}$.
The gadget state is constructed so that the qubit follows a path through this gadget which passes an inverse phase gate if and only if $\HE.\Dec_{\sk}(\encrypted{a})$ equals $1$.
The Pauli corrections can then be updated using the homomorphically-encrypted classical information and the measurement outcomes.

\subsubsection*{Specification of gadget.}
Assume $\HE.\Dec$ is computable in space logarithmic in the security parameter $\kappa$.
In Section~\ref{sec:gadget-construction} we will show that there exists an efficient algorithm $\SCHEME.\GenGadget_{\pk'}(\sk)$ which
produces a gadget: a quantum state $\Gamma_{pk'}(\sk)$ of the form as specified in this section.

The gadget will able to remove a single phase gate $\P^a$, using only knowledge of $\encrypted{a}$, where $\encrypted{a}$ decrypts to $a$ under
the secret key $\sk$. The public key $\pk'$ is used to encrypt all classical information which is part of the gadget.

The quantum part of the gadget consists of $2 m$ qubits, with $m$ some number which is polynomial in the security parameter $\kappa$. Let
$\{ (s_1, t_1), (s_2, t_2), \dots, (s_m, t_m) \}$ be disjoint pairs in $\{1, 2, \dots, 2m \}$,
and let $p \in \{0,1\}^m$ be a string of $m$ bits. Let $g(\sk)$ be a shorthand for the tuple of both of these, together with the secret key $\sk$;
\[
 g(\sk) := ( \{ (s_1, t_1), (s_2, t_2), \dots, (s_m, t_m) \} , p, \sk) \,.
\]
The tuple $g(\sk)$ is the classical information that determines the structure of the gadget as a function of the secret key $\sk$. The length of $g(\sk)$ is not dependent on the secret key: the number of qubits $m$ and 
the size of $\sk$ itself are completely determined by the choice of protocol $\HE$ and the security parameter $\kappa$.

For any bitstring $x,z\in \{0,1\}^m$, define the quantum state
\[
 \gamma_{x,z}\bigl(g(\sk)\bigr) := \prod^m_{i=1} \X^{x[i]} \Z^{z[i]} \bigl( \P^{\dagger}\bigr)^{p[i]}  \proj{\Phi^+}_{s_i t_i} \P^{p[i] } \Z^{z[i]} \X^{x[i]} \, .
\]
(Here the single-qubit gates are applied to $s_i$, the first qubit of the entangled pair.) This quantum state is a collection of maximally-entangled pairs of qubits, some with an extra inverse phase gate applied, where the pairs are determined by the disjoint pairs $\{ (s_1, t_1), (s_2, t_2), \dots, (s_m, t_m) \} $ chosen earlier. The entangled pairs have arbitrary Pauli operators applied to them, described by the bitstrings $x$ and $z$.

Note that, no matter the choice of gadget structure, averaging over all possible $x,z$ gives the completely mixed state on $2m$ qubits,
\[
 \frac{1}{2^{2m}} \sum_{\mathclap{x,z \in \{0,1\}^m}} \gamma_{x,z}\bigl(g(\sk)\bigr) = \frac{\id_{2^{2m}}}{2^{2m}} \,.
\]
This property will be important in the security proof; intuitively it shows that these gadgets do not reveal any information about $\sk$ whenever $x$ and $z$ are encrypted with a secure classical encryption scheme.

The entire gadget then is given by 
\[
 \Gamma_{\pk'}(\sk) = \rho (\HE.\Enc_{\pk'}\bigl(g(\sk)\bigr) ) \otimes \frac{1}{2^{2m}} \sum_{\mathclap{x,z \in \{0,1\}^m}} \rho (\HE.\Enc_{\pk'}(x,z)) \otimes \gamma_{x,z}\bigl(g(\sk)\bigr) \, .
\]
To summarize, the gadget consists of a quantum state $\gamma_{x,z}\bigl(g(\sk)\bigr)$, instantiated with randomly chosen $x,z$, the classical information denoting the random choice of $x,z$,
and the other classical information $g(\sk)$ which specifies the gadget. All classical information is homomorphically encrypted with a public key $\pk'$.

Since this gadget depends on the secret key $\sk$, simply encrypting this information using the public key corresponding
to $\sk$ would not be secure, unless we assume that $\HE.\Dec$ is circularly secure.
In order to avoid the requirement of circular security, we will always use a fresh, independent key $\pk'$ to encrypt this information.
The evaluator will have to do some recrypting before he is able to use this information, but otherwise using independent keys does not complicate the construction much. More details on how the evaluation procedure deals with the different keys is provided in Section~\ref{sec:evaluation}.

\subsubsection*{Usage of gadget.}
The gadget is used by performing Bell measurements between pairs of its qubits, together with an input qubit that needs a correction, without having knowledge of the structure of the gadget.

The choice of measurements can be generated by an efficient (classical) algorithm $\SCHEME.\GenMeasurement(\encrypted{a})$ which produces a list $M$ containing $m$ disjoint pairs 
of elements in $\{0,1,2,\dots,2m\}$. Here the labels $1$ to $2m$ refer to the qubits that make up a gadget and $0$ is the label of the qubit with the possible $\P$ error. The pairs represent which qubits will be connected through Bell measurements; note that all but a single qubit will be measured according to $M$.

Consider an input qubit, in some arbitrary state $\P^a \ket{\psi}$, i.e.\ the qubit has an extra phase gate if $a=1$. Let $\encrypted{a}$ be an encrypted
version of $a$, such that $a = \HE.\Dec_{\sk}(\encrypted{a})$. Then the evaluator performs Bell measurements
on $\Gamma_{\pk'}(\sk)$ and the input qubit,
according to $M \leftarrow \SCHEME.\GenMeasurement(\encrypted{a})$.
By construction,
one out the $2m+1$ qubits is still unmeasured.
This qubit will be in the state $\X^{a'} \Z^{b'} \ket{\psi}$, for some $a'$ and $b'$, both of which are functions of the specification of the gadget, the measurement
choices which depend on $\encrypted{a}$, and the outcomes of the teleportation measurements. Also see Section~\ref{sec:evaluation} and
Appendix~\ref{app:key-update-gadget} for a more in-depth explanation of how the accompanying classical information is updated.

Intuitively, the `path' the qubit takes through the gadget state, goes through one of
the fully entangled pairs with an inverse phase gate whenever $\HE.\Dec_{\sk}(\encrypted{a}) = 1$, and avoids all such pairs whenever $\HE.\Dec_{\sk}(\encrypted{a}) = 0$.

\subsection{Key generation}
Using the classical $\HE.\KeyGen$ as a subroutine to create multiple classical homomorphic keysets, we generate a classical secret and public key, and a classical-quantum evaluation key that contains $L$ gadgets, allowing evaluation of a circuit containing up to $L$ $\T$ gates.
Every gadget depends on a different secret key, and its classical information is always encrypted using the next public key.
The key generation procedure $\SCHEME.\KeyGen(1^\kappa, 1^L)$ is defined as follows:
\begin{enumerate}
\item For $i = 0$ to $L$: execute $(\pk_i, \sk_i, \evk_i) \leftarrow \HE.\KeyGen(1^{\kappa})$ to obtain $L+1$ independent classical homomorphic key sets.
\item Set the public key to be the tuple $(\pk_i)_{i = 0}^{L}$.
\item Set the secret key to be the tuple $(\sk_i)_{i = 0}^{L}$.
\item For $i = 0$ to $L-1$:  Run the procedure
$\SCHEME.\GenGadget_{\pk_{i+1}}(\sk_i)$ to create the gadget $\Gamma_{\pk_{i+1}}(\sk_i)$ as described in Section~\ref{sec:gadget}.

\item Set the evaluation key to be the set of all gadgets created in the previous step (including their encrypted classical information), plus the tuple $(\evk_i)_{i=0}^L$.
The resulting evaluation key is the quantum state
\[
\bigotimes_{i = 0}^{L-1} \left( \Gamma_{\pk_{i+1}}(\sk_i) \otimes \proj{\evk_i} \right).
\]
\end{enumerate}

\subsection{Encryption}
The encryption procedure $\SCHEME.\Enc$ is identical to $\CL.\Enc$, using the first public key $\pk_0$ for the encryption of the one-time-pad keys. We restate it here for completeness.

Every single-qubit state $\sigma$ is encrypted separately with a quantum one-time pad, and the pad key is (classically) encrypted and appended to the quantum encryption of $\sigma$, resulting in the classical-quantum state:
\[
\sum_{a,b \in \{0,1\}} \frac{1}{4}\rho(\HE.\Enc_{\pk_0}(a), \HE.\Enc_{\pk_0}(b)) \otimes X^a Z^b \sigma Z^b X^a .
\]

\subsection{Circuit evaluation}\label{sec:evaluation}
Consider a circuit with $n$ wires. The evaluation of the circuit on the encrypted data is carried out one gate at a time. 

Recall that our quantum circuit is written using a gate set that consists of the Clifford group generators $\{\H, \P, \cnot \}$ and the $\T$ gate.
A Clifford gate may affect multiple wires at the same time, while $\T$ gates can only affect a single qubit.
Before the evaluation of a single gate $U$, the encryption of an $n$-qubit state $\rho$ is of the form 
\[
\bigl(\X^{a_1} \Z^{b_1} \otimes \dots \otimes \X^{a_n} \Z^{b_n} \bigr) \rho \, \bigl( \X^{a_1} \Z^{b_1} \otimes \dots \otimes \X^{a_n} \Z^{b_n} \bigr) .
\]
The evaluating party holds the encrypted versions $\encrypted[i]{a_1}, \dots, \encrypted[i]{a_n}$ and $\encrypted[i]{b_1}, \dots, \encrypted[i]{b_n}$, with respect to the $i$th key set for some $i$ (initially, $i = 0$). The goal is to obtain a quantum encryption of the state $U\rho \, U^{\dagger}$,
such that the evaluator can homomorphically compute the encryptions of the new keys to the quantum one-time pad.
If $U$ is a Clifford gate, these encryptions will still be in the $i$th key. If $U$ is a $\T$ gate, then all encryptions are transferred to the $(i+1)$th key during the process.

\begin{itemize}
 
\item
If $U$ is a Clifford gate, we proceed exactly as in $\CL.\Eval$. The gate $U$ is simply applied to the encrypted qubit, and since $U$ commutes with the Pauli group, the evaluator only needs to update the encrypted keys in a straightforward way.
For more details, see Appendix~\ref{app:key-update-clifford}.

\item
If $U = \T$, the evaluator should start out by applying a $\T$ gate to the appropriate wire $w$. Afterwards, the qubit at wire $w$ is in the state \[ \bigl( \P^{a_w}\X^{a_w}\Z^{b_w} \T \bigr) \rho_w  \bigl( \T^\dagger  \X^{a_w}\Z^{b_w} (\P^\dagger)^{a_w} \bigr) .\]
In order to remove the $\P$ error, the evaluator uses one gadget $\Gamma_{\pk_{i+1}}(\sk_{i})$ from the evaluation key; he possesses the classical information $\encrypted[i]{a_w}$ encrypted with the correct key to be able to compute measurements $M \leftarrow \SCHEME.\GenMeasurement(\encrypted[i]{a_w})$ and performs the measurements on the pairs given by $M$.
Afterwards, using his own measurement outcomes, the classical information accompanying the gadget (encrypted using $\pk_{i+1}$), and the recryptions of $\encrypted[i]{a_w}$ and $\encrypted[i]{b_w}$ into $\encrypted[i+1]{a_w}$ and $\encrypted[i+1]{b_w}$, the evaluator homomorphically computes the new keys $\encrypted[i+1]{a'_w}$ and $\encrypted[i+1]{b'_w}$. See also Figure~\ref{fig:t-gate-evaluation} and Appendix~\ref{app:key-update-gadget}. After these computations, he should also recrypt the keys of all other wires into the $(i+1)$th key set.

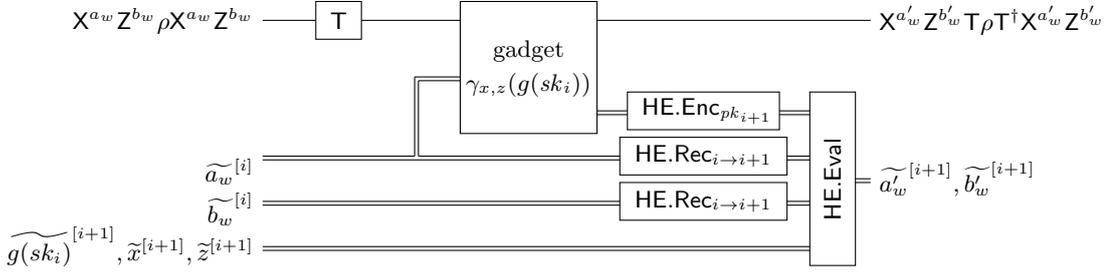
\begin{figure}[h]
\centering
\hspace*{-1cm} \begin{tikzpicture}

%%WIRES%%

%1. the qubit:
\node[anchor=east] at (0,3) {$\X^{a_w}\Z^{b_w}\rho\X^{a_w}\Z^{b_w}$};
\draw(0,3) -- (8,3);
\node[anchor=west] at (8,3) {$\X^{a'_w}\Z^{b'_w}\T\rho\T^\dag\X^{a'_w}\Z^{b'_w}$};

%2. classical key info a:
\node[anchor=east] at (0,1) {$\encrypted[i]{a_w}$};
\draw (0,1.2)   -- (2,1.2); \draw(2.05,1.2) -- (7.5,1.2);
\draw (0,1.15) -- (7.5,1.15);
\draw (2,     1.2) -- (2,     2.25);
\draw (2.05,1.2) -- (2.05,2.2);
\draw (2,     2.25) -- (3.5,2.25);
\draw (2.05,2.2)   -- (3.5,2.2);

%3. measurement outcomes:
\draw(3.5,1.8) -- (7.5,1.8);
\draw(3.5,1.75) -- (7.5,1.75);

%4. classical key info b:
\node[anchor=east] at (0,0.5) {$\encrypted[i]{b_w}$};
\draw (0,0.6) -- (7.5,0.6);
\draw (0,0.55) -- (7.5,0.55);

%5. classical gadget information:
\node[anchor=east] at (0,0) {$\encrypted[i+1]{g(sk_i)}, \encrypted[i+1]{x}, \encrypted[i+1]{z}$};
\draw(0,0) -- (7.5,0);
\draw(0,-0.05) -- (7.5,-0.05);

%6 the new encrypted keys:
\node[anchor=west] at (8,0.9) {$\encrypted[i+1]{a'_w}, \encrypted[i+1]{b'_w}$};
\draw(7.5,0.9) -- (8,0.9);
\draw(7.5,0.85) -- (8,0.85);

%%GATES%%

% the T gate:
\filldraw[fill=white](0.7,2.75) rectangle (1.3,3.25);
\node at (1,3) {$\T$};

% the gadget:
\filldraw[fill=white](2.6,1.5) rectangle (4.4,3.25);
\node at (3.5,2.6) {gadget};
\node at (3.5,2.15) {$\gamma_{x,z}(g(\sk_{i}))$};

%the encryption:
\filldraw[fill=white](4.8,1.55) rectangle (6.8,2.05);
\node at (5.85,1.8) {$\HE.\Enc_{\pk_{i+1}}$};

%the recryptions:
\filldraw[fill=white](4.7,0.95) rectangle(6.9,1.45);
\node at (5.8,1.2) {$\HE.\Rec_{i \to i+1}$};
\filldraw[fill=white](4.7,0.35) rectangle(6.9,0.85);
\node at (5.8,0.6) {$\HE.\Rec_{i \to i+1}$};

%the classical homomorphic evaluation:
\filldraw[fill=white](7.2,-0.25) rectangle (7.8,2.05);
\node[rotate=90] at (7.5,0.9) {$\HE.\Eval$};

\end{tikzpicture}
\caption{The homomorphic evaluation of the $(i+1)$th $\T$ gate of the circuit. The gadget is consumed during the process. After the use of the gadget, the evaluator encrypts his own classical information (including measurement outcomes) in order to use it in the homomorphic computation of the new keys. $\HE.\Eval$ evaluates this fairly straightforward computation that consists mainly of looking up values in a list and adding them modulo~2. Note that $\encrypted[i+1]{\sk_i}$, needed for the recryption procedures, is contained in the evaluation key.}
\label{fig:t-gate-evaluation}
\end{figure}

\end{itemize}

At the end of the evaluation of some circuit $\C$ containing $k$ $\T$ gates, the evaluator holds a one-time-pad encryption of the state $\C\ket\psi$, together with the keys to the pad, classically encrypted in the $k$th key.  The last step is to recrypt (in $L-k$ steps) this classical information into the $L$th (final) key. Afterwards, the quantum state and the key encryptions are sent to the decrypting party.

\subsection{Decryption}
The decryption procedure is identical to $\CL.\Dec$. For each qubit, $\HE.\Dec_{\sk_L}$ is run twice in order to retrieve the keys to the quantum pad. The correct Pauli operator can then be applied to the quantum state in order to obtain the desired state $\C\ket\psi$.

The decryption procedure is fairly straightforward, and its complexity does not depend on the circuit that was evaluated. This is formalized in a compactness theorem for the $\SCHEME$ scheme:

\begin{theorem}
If $\HE$ is compact, then $\SCHEME$ is compact.
\end{theorem}
\begin{proof}
Note that because the decryption only involves removing a one-time pad from the quantum ciphertext produced by the circuit evaluation,
this decryption can be carried out a single qubit at a time. 
By compactness of $\HE$, there exists a polynomial $p(\kappa)$ such that for any function $f$, the complexity of applying $\HE.\Dec$ to the output of $\HE.\Eval^f$ is at most $p(\kappa)$.
Since the keys to the quantum one-time pad of any wire $w$ are two single bits encrypted with the classical $\HE$ scheme, decrypting the keys for one wire requires at most $2p(\kappa)$ steps.
Obtaining the qubit then takes at most two steps more for (conditionally) applying $\X^{a_w}$ and $\Z^{b_w}$.
The total number of steps is polynomial in $\kappa$ and independent of $\C$, so we conclude that $\SCHEME$ is compact.
\end{proof}

\section{Security of TP}\label{sec:security}
In order to guarantee the privacy of the input data, we need to argue that an adversary that does not possess the secret key cannot learn anything about the data with more than negligible probability. To this end, we show that $\SCHEME$ is q-IND-CPA secure, i.e.\ no polynomial-time quantum adversary can tell the difference between an encryption of a real message and an encryption of $\proj{0}$, even if he gets to choose the message himself (recall the definition of q-IND-CPA security from Section~\ref{sec:prelim-security}). Like in the security proofs in~\cite{BJ15}, we use a reduction argument to relate the probability of being able to distinguish between the two encryptions to the probability of winning an indistinguishability experiment for the classical $\HE$, which we already know to be small. The aim of this section is to prove the following theorem:

\begin{theorem}\label{thm:security}
If $\HE$ is q-IND-CPA secure, then $\SCHEME$ is q-IND-CPA secure for circuits containing up to polynomially (in $\kappa$) many $\T$ gates.
\end{theorem}

In order to prove Theorem~\ref{thm:security}, we first prove that an efficient adversary's performance in the indistinguishability game is only negligibly different whether or not he receives a real evaluation key with real gadgets, or just a completely mixed quantum state with encryptions of 0's accompanying them (Corollary~\ref{cor:security}). Then we argue that without the evaluation key, an adversary does not receive more information than in the indistinguishability game for the scheme $\CL$, which has already been shown to be q-IND-CPA secure whenever $\HE$ is.

We start with defining a sequence of variations on the $\SCHEME$ scheme. For $\ell \in \{0, \dots, L\}$, let $\SCHEME^{(\ell)}$ be identical to $\SCHEME$, except for the key generation procedure: $\SCHEME^{(\ell)}.\KeyGen$ replaces, for every $i \geq \ell$, all classical information accompanying
the $i$th gadget with the all-zero string before encrypting it. For any number $i$, define the shorthand 
\[
g_i := g(\sk_i).  
\]
As seen in Section~\ref{sec:gadget}, the length of the classical information does not depend on $\sk_i$ itself, so a potential adversary cannot gain any information about $\sk_i$ just from this encrypted string. In summary,
\begin{align*}
\SCHEME^{(\ell)}.\KeyGen(1^\kappa, 1^L) &:= \bigotimes_{i=0}^{L-1} \proj{\evk_i} \otimes \bigotimes_{i = 0}^{\ell-1} \Gamma_{\pk_{i+1}}(\sk_i) \otimes 
                                                                                 \\
                                                     &\phantom{:= \ }\bigotimes_{i = \ell}^{L-1}  \Bigl(  \rho( \HE.\Enc_{\pk_{i+1}}(0^{|g_i|}) ) \otimes \\
                                                     &\phantom{:= \  \bigotimes_{i = \ell}^{L-1}  \Bigl( } \frac{1}{2^{2m}} \enskip \sum_{ \mathclap{ x, z \in \{0,1\}^m } } \, \rho(\HE.\Enc_{\pk_{i+1}}(0^{m},0^{m}) ) \otimes \gamma_{x,z}(g_i) \Bigr) \, . \\
\end{align*}

Intuitively, one can view $\SCHEME^{(\ell)}$ as the scheme that provides only $\ell$ usable gadgets in the evaluation key.
Note that $\SCHEME^{(L)} = \SCHEME$, and that in $\SCHEME^{(0)}$, only the classical evaluation keys remain, since without the encryptions of the classical $x$ and $z$, the quantum part
of the gadget is just the completely mixed state. That is, we can rewrite the final line of the previous equation as
\begin{align}
 \frac{1}{2^{2m}} \enskip \sum_{ \mathclap{ x, z \in \{0,1\}^m } } \,& \rho(\HE.\Enc_{\pk_{i+1}}(0^{m},0^{m}) ) \otimes \gamma_{x,z}(g_i) \nonumber \\
\\ &= \rho(\HE.\Enc_{\pk_{i+1}}(0^{m},0^{m}) ) \otimes \frac{\id_{2^{2m}} }{2^{2m}} \, \label{eq:mixedgadget}  .
\end{align}

With the definitions of the new schemes, we can lay out the steps to prove Theorem~\ref{thm:security} in more detail. First, we show that in the quantum CPA indistinguishability experiment, any efficient adversary interacting with $\SCHEME^{(\ell)}$ only has negligible advantage over an adversary interacting with $\SCHEME^{(\ell-1)}$, i.e.~the scheme where the classical information $g_{\ell-1}$ is removed (Lemma~\ref{lem:security-step1}). 
By iteratively applying this argument, we are able to argue that the advantage of an adversary who interacts with $\SCHEME^{(L)}$
over one who interacts with $\SCHEME^{(0)}$ is also negligible (Corollary~\ref{cor:security}).
Finally, we conclude the proof by arguing that $\SCHEME^{(0)}$ is q-IND-CPA secure by comparison to the $\CL$ scheme.

\begin{lemma}\label{lem:security-step1}
Let $0 < \ell \leq L$. If $\HE$ is q-IND-CPA secure, then for any quantum polynomial-time adversary $\advA = (\advA_1, \advA_2)$, there exists a negligible function $\eta$ such that
\[
\Pr[\PubK{\advA, \SCHEME^{(\ell)}} = 1] - \Pr[\PubK{\advA, \SCHEME^{(\ell-1)}} = 1] \leq \eta(\kappa).
\]
\end{lemma}
\begin{proof}
The difference between schemes $\SCHEME^{(\ell)}$ and $\SCHEME^{(\ell-1)}$ lies in whether the gadget state $\gamma_{x_{\ell-1},z_{\ell-1}}(g_{\ell-1})$ is supplemented with its classical information $\encrypted{g_{\ell-1}}, \encrypted{x_{\ell-1}},\encrypted{z_{\ell-1}}$, or just with an encryption of $0^{\lvert g_{\ell-1}\rvert + 2m}$.

Let $\advA = (\advA_1, \advA_2)$ be an adversary for the game $\PubK{\advA, \SCHEME^{(\ell)}}$. We will define an adversary $\advA' = (\advA'_1, \advA'_2)$ for $\PubKm{\advA', \HE}$ that will either simulate the game $\PubK{\advA, \SCHEME^{(\ell)}}$ or $\PubK{\advA,\SCHEME^{(\ell-1)}}$.
Which game is simulated will depend on some $s \in_R \{0,1\}$ that is unknown to $\advA'$ himself. Using the assumption that $\HE$ is q-IND-CPA secure, we are able to argue that $\advA'$ is unable to recognize which of the two schemes was simulated, and this fact allows us to bound the difference in success probabilities between the security games of $\SCHEME^{(\ell)}$ and $\SCHEME^{(\ell-1)}$. The structure of this proof is very similar to e.g.~Lemma~5.3 in \cite{BJ15}. The adversary $\advA'$ acts as follows (see also Figure~\ref{fig:security}):
\begin{description}
\item[$\advA'_1$] takes care of most of the key generation procedure: he generates the classical key sets 0 through $\ell-1$ himself, 
generates the random strings $x_0,z_0, \dots,$ $x_{\ell-1},z_{\ell-1}$, and constructs the gadgets $\gamma_{x_0,z_0}(g_0), \dots, \gamma_{x_{\ell-1},z_{\ell-1}}(g_{\ell-1})$ and their classical information $g_0, \dots, g_{\ell-1}$. He encrypts the classical information using the appropriate public keys. Only $g_{\ell-1}$, $x_{\ell-1}$ and $z_{\ell-1}$ are left unencrypted: instead of encrypting these strings himself using $\pk_{\ell}$, $\advA'_1$ sends the strings for encryption to the challenger. Whether the challenger really encrypts $g_{\ell-1}$, $x_{\ell-1}$ and $z_{\ell-1}$ or replaces the strings with a string of zeros, determines which of the two schemes is simulated. $\advA'$ is unaware of the random choice of the challenger.

The adversary $\advA'_1$ also generates the extra padding inputs that correspond to the already-removed gadgets $\ell$ up to $L-1$.
Since these gadgets consist of all-zero strings encrypted with independently chosen public keys that are not used anywhere else, together with a completely mixed quantum state (as shown in Equation~\ref{eq:mixedgadget}), the adversary can generate them without needing any extra information.
\item[$\advA'_2$] feeds the evaluation key and public key, just generated by $\advA'_1$, to $\advA_1$ in order to obtain a chosen message $\mathcal{M}$ (plus the auxiliary state $\mathcal{E}$). He then picks a random $r \in_R \{0,1\}$ and erases $\mathcal{M}$ if and only if $r = 0$. He encrypts the result according to the $\SCHEME.\Enc$ procedure (using the public key $(\pk_i)_{i=0}^{L}$ received from $\advA'_1$), and gives the encrypted state, plus $\mathcal{E}$, to $\advA_2$, who outputs $r'$ in an attempt to guess $r$. $\advA'_2$ now outputs 1 if and only if the guess by $\advA$ was correct, i.e.~$r \equiv r'$.
\end{description}

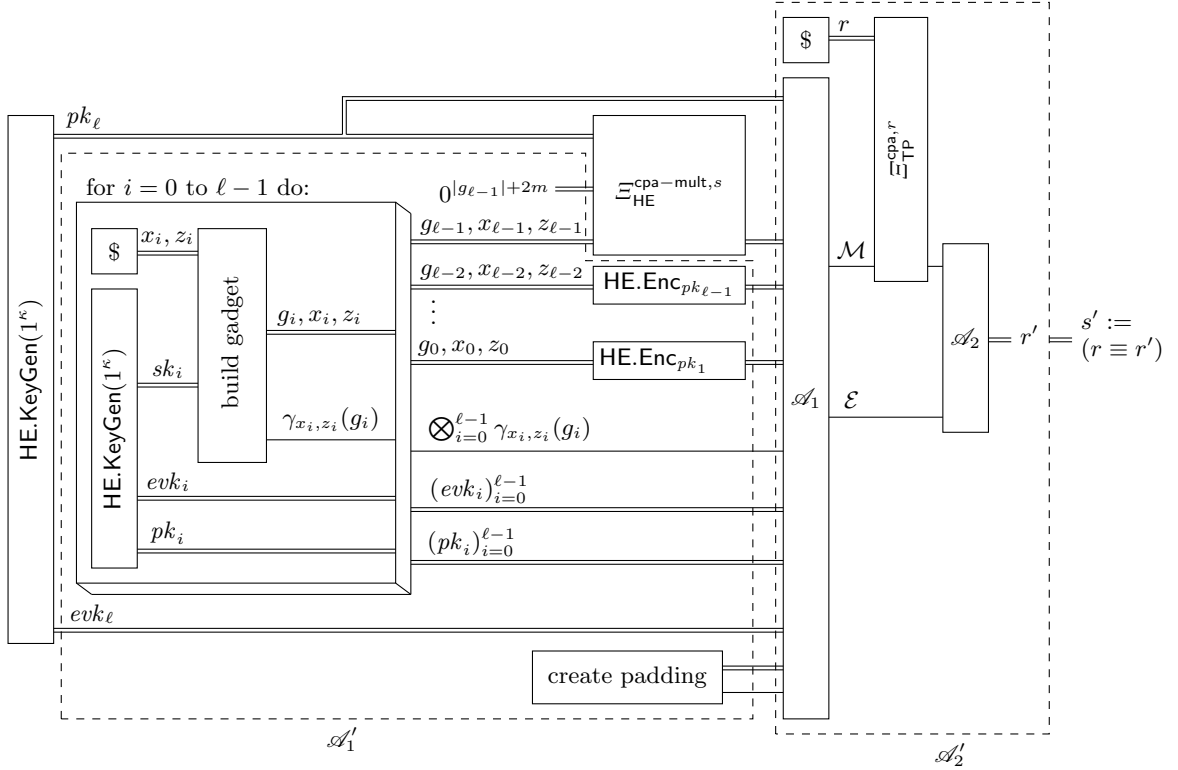
\begin{figure}[h]
\centering
\makebox[\textwidth][c]{\begin{tikzpicture}

%\def\meter(#1)(#2)
 % { \filldraw[fill=white] (9.75,9.8) rectangle (10.25,10.2);
%\draw (10.15,9.9) arc (45:135:0.25);
%\draw (10,9.9) -- (10.1,10.1); }

%keygen:
\draw(1,5.75) -- (5,5.75);
\draw(5.05,5.75) -- (9,5.75);
\draw(5,5.75) -- (5,6.25) -- (11.1,6.25);
\draw(5.05,5.75) -- (5.05,6.2) -- (11.1,6.2);
\draw(1,5.7) -- (9,5.7);
\node[anchor=south] at (1.6,5.7) {$\pk_\ell$};

\draw(1,-0.8) -- (11.1,-0.8);
\draw(1,-0.85) -- (11.1,-0.85);
\node[anchor=south] at (1.7,-0.85) {$\evk_\ell$};

\filldraw[fill=white] (0.6,-1) rectangle (1.2, 6);
\node[rotate=90] at (0.9,2.5) {$\HE.\KeyGen(1^{\kappa})$};

%own keygen:
\draw(2,0.25)   -- (5.7,0.25);
\draw(2,0.2) -- (5.7,0.2);
\node[anchor=south] at (2.7,0.2) {$\pk_i$};
\draw(2,0.95)   -- (5.7,0.95);
\draw(2,0.9) -- (5.7,0.9);
\node[anchor=south] at (2.7,0.9) {$\evk_i$};
\draw(2,2.45)   -- (3.5,2.45);
\draw(2,2.4) -- (3.5,2.4);
\node[anchor=south] at (2.7,2.4) {$\sk_i$};
\filldraw[fill=white] (1.7,0) rectangle (2.3,3.7);
\node[rotate=90] at (2,1.85) {$\HE.\KeyGen(1^{\kappa})$};

%randomness inside loop:
%\draw (2,4.2) -- (2.8,4.2) -- (2.8,3.2) -- (3.5,3.2);
%\draw (2,4.15) -- (2.75,4.14) -- (2.75, 3.15) -- (3.5,3.15);
\draw (2,4.2) -- (3.5,4.2);
\draw (2,4.15) -- (3.5,4.15);
\filldraw[fill=white] (1.7,3.9) rectangle (2.3,4.5);
\node at (2,4.2) {$\$$};
\node[anchor=south] at (2.7,4.15) {$x_i,z_i$};

%gadget construction:
\draw(4,1.7) -- (5.7,1.7);
\node[anchor=south] at (4.85,1.7) {$\gamma_{x_i,z_i}(g_i)$};
\draw(4,3.15) -- (5.7,3.15);
\draw(4,3.1) -- (5.7,3.1);
\node[anchor=south] at (4.7,3.1) {$g_i,x_i,z_i$};

\filldraw[fill=white] (3.1,1.4) rectangle (4,4.5);
\node[rotate=90] at (3.55, 3) {build gadget};
%\node at (3.9,2.65) {construct};
%\node at (3.9,2.15) {gadget};

%loop:
\draw (1.5,-0.2) rectangle (5.7,4.85);
\draw (5.9,-0.35) -- (5.9,4.7);
\draw(5.7,4.85) -- (5.9,4.7);
\draw(5.7,-0.2) -- (5.9,-0.35);
\draw(1.7,-0.35) -- (5.9, -0.35);
\draw(1.5,-0.2) -- (1.7,-0.35);
\node[anchor=south] at (3.1,4.8) {for $i = 0$ to $\ell-1$ do:};

%loop results:
\draw (5.9,4.35) -- (11.1,4.35);
\draw (5.9,4.3) -- (11.1,4.3);
\node[anchor=south] at (7.1,4.3) {$g_{\ell-1}, x_{\ell-1}, z_{\ell-1}$};

\draw (5.9,3.75) -- (11.1,3.75);
\draw (5.9,3.7) -- (11.1,3.7);
\node[anchor=south] at (7.1,3.7) {$g_{\ell-2},x_{\ell-2},z_{\ell-2}$};

\node at (6.2,3.5) {$\vdots$};

\draw (5.9,2.75) -- (11.1,2.75);
\draw (5.9,2.7) -- (11.1,2.7);
\node[anchor=south] at (6.6,2.7) {$g_0,x_0,z_0$};

\draw (5.9,1.55) -- (11.1,1.55);
\node[anchor=south] at (7.2,1.55) {$\bigotimes_{i=0}^{\ell-1} \gamma_{x_i,z_i}(g_i)$};

\draw (5.9,0.8) -- (11.1, 0.8);
\draw (5.9,0.75) -- (11.1,0.75);
\node[anchor=south] at (6.8,0.75) {$(\evk_i)_{i=0}^{\ell-1}$};

\draw (5.9,0.1) -- (11.1, 0.1);
\draw (5.9,0.05) -- (11.1,0.05);
\node[anchor=south] at (6.7,0.05) {$(\pk_i)_{i=0}^{\ell-1}$};

%encrypt:
\filldraw[fill=white] (8.3,2.5) rectangle (10.3,3);
\node at (9.1,2.75) {$\HE.\Enc_{\pk_1}$};

\filldraw[fill=white] (8.3,3.5) rectangle (10.3,4);
\node at (9.3,3.75) {$\HE.\Enc_{\pk_{\ell-1}}$};

%xi test:
\draw (7.8,5) -- (8.8,5);
\draw (7.8,5.05) -- (8.8,5.05);
\node at (7,5) {$0^{\lvert g_{\ell-1} \rvert + 2m}$};
\filldraw[fill=white] (8.3,4.15) rectangle (10.3,6);
\node at (9.3,5) {$\Xi_{\HE}^{\mathsf{cpa-mult}, s}$};
%\node[anchor=south] at (9.7,4.3) {$c$};

%padding:
\draw (9,-1.3) -- (11.3,-1.3);
\draw (9,-1.35) -- (11.3,-1.35);
\draw (9,-1.65) -- (11.3,-1.65);
\filldraw[fill=white] (7.5,-1.8) rectangle (10,-1.1);
\node at (8.75,-1.45) {create padding};

%advA1:
\draw (11.1,4) -- (13,4);
\node[anchor=south] at (11.7,4) {$\mathcal{M}$};

\draw (11.1, 2) -- (13,2);
\node[anchor=south] at (11.7,2) {$\mathcal{E}$};
\filldraw[fill=white] (10.8,-2) rectangle (11.4,6.5);
\node at (11.1,2.25) {$\advA_1$};

%erase message:
\draw(11.2,7.05) -- (12.2,7.05);
\draw(11.2,7) -- (12.2,7);
\node[anchor=south] at (11.6,7) {$r$};
\filldraw[fill=white] (10.8,6.7) rectangle (11.4,7.3);
\node at (11.1,7) {$\$$};

\filldraw[fill=white] (12,3.8) rectangle (12.7,7.3);
\node[rotate=90] at (12.35,5.55) {$\Xi^{\mathsf{cpa},r}_{\SCHEME}$};

%erase message:
%\draw (11.3,3.6) rectangle (12.5,4.4);
%\meter{11.6}{4}
%\node at (12.2,4) {$\ket0$};

%\draw(11.9,5.75) -- (11.9,4.8);
%\draw(11.95,5.75) -- (11.95,4.8);
%\node[anchor=west] at (11.9,5.2) {$r$};
%\node[anchor=south] at (11.9,4.35) {if $r = 0$:};
%\filldraw[fill=white] (11.6,5.5) rectangle (12.2,6);
%\node at (11.9,5.75) {$\$$};

%q-encrypt:
%\filldraw[fill=white] (12.7,3.7) rectangle (14.3,4.3);
%\node at (13.5,4) {$\SCHEME.\Enc_{\pk}$};

%advA2:
\draw (13.3,3.05) -- (13.8,3.05);
\draw (13.3,3) -- (13.8,3);
\node[anchor=west] at (13.8,3.1) {$r'$};
\filldraw[fill=white] (12.9,1.8) rectangle (13.5,4.3);
\node at (13.2,3.05) {$\advA_2$};

%adv':
\draw[dashed] (1.3,-2) -- (1.3,5.5) -- (8.2,5.5) -- (8.2,4.075) -- (10.4,4.075) -- (10.4,-2) -- (1.3,-2);
\node[anchor=north] at (5,-2) {$\advA'_1$};

\draw[dashed] (10.7,-2.2) -- (10.7,7.5) -- (14.3,7.5) -- (14.3,-2.2) -- (10.7,-2.2);
\node[anchor=north] at (13,-2.2) {$\advA'_2$};

%final answer:
\draw (14.3,3.05) -- (14.6,3.05);
\draw (14.3,3) -- (14.6,3);
\node[anchor=west] at (14.6,3.3) {$s' :=$};
\node[anchor=west] at (14.6,2.9) {$(r \equiv r')$};

%\filldraw[fill=white] (6.7,4.75) rectangle (7.3,5.25);

%\filldraw[fill=white] (8.5,4) rectangle (10.5,4.5);
%\node at (9.3,4.25) {$\HE.\Enc_{\pk_{L}}$};

%MIDDLE = 10,10
%\filldraw[fill=white] (9.75,9.8) rectangle (10.25,10.2);
%\draw (10.15,9.9) arc (45:135:0.25);
%\draw (10,9.9) -- (10.1,10.1);

%inside throw gadgets away:
%throw gadgets away (box):
%\node[anchor=south] at (6.6,2.45) {iff $i \geq k$:};
%\filldraw[fill=white] (6,10) rectangle (8,11.5);
%\draw (6,11) -- (6.5,11);
%\draw (7.5,11) -- (8,11);
%\node[anchor=east] at (7.5,11) {$\id_d$};
%\meter{6.5}{11}
%\draw (6,10.25) -- (6.5,10.25);
%\draw (6,10.2) -- (6.5,10.2);
%\draw (6.5,10.1) -- (6.5,10.35);
%\draw (7.5,10.25) -- (8,10.25);
%\draw(7.5,10.2)    -- (8,10.2);
%\node[anchor=east] at (7.6,10.25) {$0^{\lvert g_i \rvert}$};

\end{tikzpicture}}
\caption{A strategy for the game $\PubKm{\advA', \HE}$, using an adversary $\advA$ for  $\PubK{\advA,\SCHEME^{(\ell)}}$ as a subroutine. All the wires that form an input to $\advA_1$ together form the evaluation key and public key for $\SCHEME^{(\ell)}$ or $\SCHEME^{(\ell-1)}$, depending on $s$.
Note that $\Xi^{\mathsf{cpa},r}_{\SCHEME} = \Xi^{\mathsf{cpa},r}_{\SCHEME^{(\ell)}} = \Xi^{\mathsf{cpa},r}_{\SCHEME^{(\ell-1)}}$, so $\advA'_2$ can run either one of these independently of $s$ (i.e.~without having to query the challenger). The `create padding' subroutine generates dummy gadgets for $\ell$ up to $L-1$, as described in the definition of $\advA_1$.}
\label{fig:security}
\end{figure}

\noindent Because $\HE$ is q-IND-CPA secure, the probability that $\advA'$ wins $\PubKm{\advA', \HE}$, i.e.~that $s' \equiv s$, is at most $\frac{1}{2} + \eta'(\kappa)$ for some negligible function $\eta'$. There are two scenarios in which $\advA'$ wins the game:
\begin{itemize}
\item $s = 1$ \emph{and} $\advA$ guesses $r$ correctly: If $s = 1$, the game that is being simulated is $\PubK{\advA, \SCHEME^{(\ell)}}$. If $\advA$ wins the simulated game ($r \equiv r'$), then $\advA'$ will correctly output $s' = 1$. (If $\advA$ loses, then $\advA'$ outputs 0, and loses as well).
\item $s = 0$ \emph{and} $\advA$ does not guess $r$ correctly: If $s = 0$, the game that is being simulated is $\PubK{\advA, \SCHEME^{(\ell-1)}}$. If $\advA$ loses the game ($r \not\equiv r'$), then $\advA'$ will correctly output $s' = 0$. (If $\advA$ wins, then $\advA'$ outputs 1 and loses).
\end{itemize}
From the above, we conclude that
\begin{outdent}
\begin{alignat*}{4}
&& \Pr[s=1]\cdot \Pr[\PubK{\advA, \SCHEME^{(\ell)}} = 1] &+ \Pr[s=0]\cdot \Pr[\PubK{\advA, \SCHEME^{(\ell-1)}} = 0] &&\leq \frac{1}{2} + \eta'(\kappa)\\
\Leftrightarrow && \frac{1}{2} \Pr[\PubK{\advA, \SCHEME^{(\ell)}} = 1] &+ \frac{1}{2} \left(1 - \Pr[\PubK{\advA, \SCHEME^{(\ell-1)}} = 1]\right) &&\leq \frac{1}{2} + \eta'(\kappa)\\
\Leftrightarrow && \Pr[\PubK{\advA, \SCHEME^{(\ell)}} = 1] &- \Pr[\PubK{\advA, \SCHEME^{(\ell-1)}} = 1] && \leq 2\eta'(\kappa)
\end{alignat*}
\end{outdent}
Set $\eta(\kappa) := 2\eta'(\kappa)$, and the proof is complete.
\end{proof}

By applying Lemma~\ref{lem:security-step1} iteratively, $L$ times in total, we can conclude that the difference between $\SCHEME^{(L)}$ and $\SCHEME^{(0)}$ is negligible, because the sum of polynomially many negligible functions is still negligible:
\begin{corollary}\label{cor:security}
If $L$ is polynomial in $\kappa$, then for any quantum polynomial-time adversary $\advA = (\advA_1, \advA_2)$, there exists a negligible function $\eta$ such that
\[
\Pr[\PubK{\advA, \SCHEME^{(L)}} = 1] - \Pr[\PubK{\advA, \SCHEME^{(0)}} = 1] \leq \eta(\kappa).
\]
\end{corollary}

Using Corollary~\ref{cor:security}, we can finally prove the q-IND-CPA security of our scheme $\SCHEME = \SCHEME^{(L)}$.

\begin{proof}[Proof of Theorem~\ref{thm:security}]
The scheme $\SCHEME^{(0)}$ is very similar to $\CL$ in terms of its key generation and encryption steps. The evaluation key consists of several classical evaluation keys, plus some completely mixed states and encryptions of 0 which we can safely ignore because they do not contain any information about the encrypted message. In both schemes, the encryption of a qubit is a quantum one-time pad together with the encrypted keys. The only difference is that in $\SCHEME^{(0)}$, the public key and evaluation key form a tuple containing, in addition to $\pk_0$ and $\evk_0$ which are used for the encryption of the quantum one-time pad, a list of public/evaluation keys that are independent of the encryption. These keys do not provide any advantage (in fact, the adversary could have generated them himself by repeatedly running $\HE.\KeyGen(1^{\kappa}, 1^L)$). Therefore, we can safely ignore these keys as well.

In \cite[Lemma 5.3]{BJ15}, it is shown that $\CL$ is q-IND-CPA secure. Because of the similarity between $\CL$ and $\SCHEME^{(0)}$, the exact same proof shows that $\SCHEME^{(0)}$ is q-IND-CPA secure as well, that is, for any $\advA$ there exists a negligible function $\eta'$ such that
\[
\Pr[\PubK{\advA, \SCHEME^{(0)}} = 1] \leq \frac{1}{2} + \eta'(\kappa).
\]
Combining this result with Corollary~\ref{cor:security}, it follows that
\begin{align*}
\Pr[\PubK{\advA, \SCHEME} = 1] &\leq \Pr[\PubK{\advA, \SCHEME^{(0)}} = 1] + \eta(\kappa)\\
&\leq \frac{1}{2} + \eta'(\kappa) + \eta(\kappa).
\end{align*}
Since the sum of two negligible functions is itself negligible, we have proved Theorem~\ref{thm:security}.
\end{proof}

\subsection{Circuit privacy}\label{sec:circuit-privacy}
The scheme \SCHEME as presented above ensures the privacy of the input data. It does not guarantee, however, that whoever generates the keys, encrypts, and decrypts cannot gain information about the circuit $\C$ that was applied to some input $\rho$ by the evaluator. Obviously, the output value $\C\rho\C^{\dag}$  often reveals something about the circuit $\C$, but apart from this necessary leakage of information, one may require a (quantum) homomorphic encryption scheme to ensure \emph{circuit privacy} in the sense that an adversary cannot statistically gain any information about $\C$ from the output of the evaluation procedure that it could not already gain from $\C\rho\C^\dag$ itself.

We claim that circuit privacy for \SCHEME in the semi-honest setting (i.e.\ against passive adversaries\footnote{Note that there various ways to define passive adversaries in the quantum setting~\cite{DNS10,BB14}. Here, we are considering adversaries that follow all protocol instructions exactly.}) can be obtained by modifying the scheme only slightly, given that the classical encryption scheme has the circuit privacy property.

\begin{restatable}{theorem}{circuitprivacy}
\label{thm:circuit-privacy}
If $\HE$ has circuit privacy in the semi-honest setting, then $\SCHEME$ can be adapted to a quantum homomorphic encryption scheme with circuit privacy.
\end{restatable}

\begin{proof-sketch}
If the evaluator randomizes the encryption of the output data by applying a quantum one-time pad to the (already encrypted) result of the evaluation, the keys themselves are uniformly random and therefore do not reveal any information about what circuit was evaluated. The evaluator can then proceed to update the classical encryptions of those keys accordingly, and by the circuit privacy of the classical scheme, the resulting encrypted keys will also contain no information about the computations performed. A more thorough proof is given in Appendix~\ref{app:circuit-privacy}.
\end{proof-sketch}

\section{Constructing the gadgets}\label{sec:gadget-construction}
In this section we will first show how to construct gadgets that have polynomial size whenever the scheme $\HE$ has a decryption circuit with logarithmic depth (i.e., the decryption function is in $\NC^1$). 
This construction will already be powerful enough to instantiate $\SCHEME$ with current classical schemes for homomorphic encryption,
since these commonly have low-depth decryption circuits. 
Afterwards, in Section~\ref{sec:gadget-construction-logspace}, we will present a larger toolkit to construct gadgets, which is  efficient for a larger class of possible decryption functions.
To illustrate these techniques, we apply these tools to create
gadgets for schemes that are based on Learning With Errors (LWE).
Finally, we will reflect on the possibility of constructing these gadgets in scenarios where quantum power is limited.

\subsection{For log-depth decryption circuits}\label{sec:barrington}
The main tool for creating gadgets that encode log-depth decryption circuits comes from Barrington's theorem:
a classic result in complexity theory, which states that all boolean circuits of logarithmic depth can be encoded as polynomial-sized width-5 permutation branching programs. Every instruction of such a branching program will be encoded as connections between five Bell pairs.

\begin{definition}
 A \emph{width-$k$ permutation branching program} of length $L$ on an input $x \in \{0,1\}^n$ is a list of $L$ instructions of the form $\langle i_\ell, \sigma^1_\ell,\sigma^0_\ell \rangle$, for $1 \leq \ell \leq L$, such that
 $i_\ell \in [n]$, and $\sigma^1_\ell$ and $\sigma^0_\ell$ are elements of $S_{k}$, i.e., permutations of $[k]$. 
 The program is executed by composing the permutations given by the instructions 1 through $L$, selecting $\sigma^1_\ell$ if $x_{i_\ell} = 1$ and selecting $\sigma^0_\ell$ if $x_{i_\ell} = 0$.
 The program \emph{rejects} if this product equals the identity permutation and \emph{accepts} if it equals a fixed $k$-cycle.
\end{definition}

Since these programs have a very simple form, it came as a surprise when they were proven to be quite powerful~\cite{Bar89}.

\begin{theorem}[Barrington~\cite{Bar89}]
Every fan-in 2 boolean circuit~$C$ of depth~$d$ can be simulated by a width-5 permutation branching program
of length at most~$4^d$.
\end{theorem}

Our gadget construction will consist of first transforming the decryption function $\HE.\Dec$ into a permutation branching program, and then
encoding this permutation branching program as a specification of a gadget,
as produced by $\SCHEME.\GenGadget_{\pk'}(\sk)$, and
 usage instructions $\SCHEME.\GenMeasurement(\encrypted{a})$.

\begin{theorem}\label{thm:gadget-logdepth}
Let $\HE.\Dec_{\sk}(\encrypted{a})$ be the decryption function of the classical homomorphic encryption scheme $\HE$. If $\HE.\Dec$ is computable by a boolean fan-in 2
circuit of depth $O(\log(\kappa))$, where $\kappa$ is the security parameter,
then there exist gadgets for $\SCHEME$ of size polynomial in $\kappa$.
\end{theorem}

\begin{proof}
Our description will consist of three steps. First, we write $\HE.\Dec$
as a width-5 permutation branching program, of which the
instructions alternately depend on the secret key~$\sk$ and on the ciphertext~$\encrypted{a}$. Secondly, we specify how to transform these instructions into
a gadget which almost works correctly, but for which the qubit
ends up at an unknown location.
Finally, we complete the construction by executing the inverse
program, so that the qubit ends up at a known location.

The first part follows directly from Barrington's theorem. The effective input of $\HE.\Dec$ can be seen as the concatenation of the secret key $\sk$ and the ciphertext $\encrypted{a}$.
Since by assumption the circuit is of depth $O(\log \kappa)$, there exists 
width-5 permutation branching program $\mathcal{P}$ of length $L = \kappa^{O(1)}$, with the following properties.
We write \[
\mathcal{P} = \left( \langle i_1, \sigma^1_1, \sigma^0_1 \rangle, 
\langle i_{2}, \sigma^1_2, \sigma^0_2 \rangle, 
\dots,
\langle i_{L}, \sigma^1_L, \sigma^0_L \rangle
\right)
\] as the list of instructions of the width-5 permutation branching program. Without loss of generality\footnote{This can be seen by inserting dummy instructions that always perform the identity permutation between any two consecutive instructions
that depend on the same variable. Alternatively, it would be possible to improve the construction by `multiplying out' consecutive instructions whenever they depend on the same variable.}, we can assume that the instructions alternately depend on bits
of $\encrypted{a}$ and bits of $\sk$. That is, the index $i_\ell$ refers to a bit of $\encrypted{a}$ if $\ell$ is odd, and to a bit of $\sk$ if $\ell$ is even. There are $L$ instructions in total, of which $L/2$ are odd-numbered and $L/2$ are even.

The output of $\SCHEME.\GenGadget_{\pk'}(\sk)$, i.e., the list of pairs that defines the structure of the gadget, will be created from the even-numbered instructions, evaluated using the secret key $\sk$. 
For every even-numbered $\ell \leq L$, we connect
ten qubits in the following way.
Suppose the $\ell^{\text{th}}$ instruction evaluates to some permutation $\sigma_\ell := \sigma^{\sk_{i_\ell}}_\ell$. Label the 10 qubits of this part of the gadget by $1_{\ell,\mathrm{in}}, 2_{\ell,\mathrm{in}}, \dots, 5_{\ell,\mathrm{in}}$ and $1_{\ell,\mathrm{out}}, 2_{\ell,\mathrm{out}}, \dots, 5_{\ell,\mathrm{out}}$. These will correspond to 5 EPR pairs, connected according to the permutation: $(1_{\ell,\mathrm{in}},
\sigma_{\ell}(1)_{\ell,\mathrm{out}})$, $(2_{\ell,\mathrm{in}}, 
\sigma_{\ell}(2)_{\ell,\mathrm{out}})$, etc., up to
$(5_{\ell,\mathrm{in}}, 
\sigma_{\ell}(5)_{\ell,\mathrm{out}})$.

After the final instruction of the branching program, $\sigma_L$,
also perform an inverse phase gate $\P^{\dag}$ on the qubits 
labeled as $2_{L,\mathrm{out}}$, $3_{L,\mathrm{out}}$,
$4_{L,\mathrm{out}}$, $5_{L,\mathrm{out}}$. Execution of
the gadget will teleport the qubit through one of these
whenever $\encrypted{a} = 1$.

For this construction, $\SCHEME.\GenMeasurement(\encrypted{a})$ will be 
given by the odd instructions, which depend on the bits of $\encrypted{a}$.
Again, for all odd $\ell \leq L$, let $\sigma_\ell := \sigma_\ell^{\encrypted{a}_{i_\ell}}$ 
be the permutation given by the evaluation of
instruction $\ell$ on $\encrypted{a}$.
For all $\ell$ strictly greater than one, the measurement instructions will be: perform a Bell measurement
according to the permutation $\sigma_\ell$ between the `out'
qubits of the previous set, and the `in' qubits of the next.
The measurement pairs will then be $(1_{\ell-1,\mathrm{out}}, \sigma(1)_{\ell,\mathrm{in}})$, $(2_{\ell-1,\mathrm{out}}, \sigma(2)_{\ell,\mathrm{in}})$, up to $(5_{\ell-1,\mathrm{out}}, \sigma(5)_{\ell,\mathrm{in}})$.

For $\ell=1$, there is no previous layer to connect to, only the input qubit. For that, we add the measurement instruction
$(0, \sigma(1)_{1,\mathrm{in}})$, where 0 is the label of the
input qubit.

By Barrington's theorem, if $\HE.\Dec_\sk(\encrypted{a}) = 0$
then the product, say $\tau$, of the permutations coming from the evaluated instructions equals the identity. In that case, consecutively
applying these permutations on `1', 
results in the unchanged starting value, `1'. If instead the decryption would output 1, the consecutive application results in another value in $\{2,3,4,5\}$, because in that case, $\tau$ is a $k$-cycle. 
After teleporting a qubit through these EPR pairs,
with teleportation measurements chosen accordingly,
the input qubit will be present at $\tau(1)_{L,\mathrm{out}}$, with an inverse phase gate if $\tau(1)$ is unequal to 1.

The gadget
constructed so far would correctly apply the phase gate,
conditioned on $\HE.\Dec_{\sk}(\encrypted{a})$,
with one problem: afterward, the qubit is at a location unknown
to the user of the gadget, because the user cannot compute $\tau$.

We fix this problem in the following way: execute the
inverse branching program afterwards. The entire construction is continued in the same way, but the instructions of the inverse program are used. The inverse program can be made from the original program by reversing the order of instructions, and then for each permutation using its inverse permutation instead.
The first inverse instruction is $\langle i_{L}, (\sigma^1_L)^{-1}, (\sigma^0_L)^{-1} \rangle$, then $\langle i_{L-1}, (\sigma^1_{L-1})^{-1}, (\sigma^0_{L-1})^{-1} \rangle$, 
with final instruction $\langle i_1, (\sigma^1_1)^{-1}, (\sigma^0_1)^{-1} \rangle$.
One small detail is that $i_{L}$
is used twice in a row, breaking the alternation; the user
of the gadget can simply perform the measurements that correspond
to the identity permutation $e$ in between, since $(\sigma^0_L)(\sigma^0_L)^{-1} = (\sigma^1_L)(\sigma^1_L)^{-1} = e$.

After having repeated the construction with the inverse permutation
branching program, the qubit is guaranteed to be at the location
where it originally started: $\sigma_1(1)$ of the final layer of five qubits -- that will then be the corrected qubit
which is the output of the gadget.

The total number of qubits which form the gadget,
created from a width-5 permutation branching program of length $L$, of which the instructions alternate between depending on $\encrypted{a}$ and depending on $\sk$, is $2\cdot(5L) = 10L$.
\end{proof}

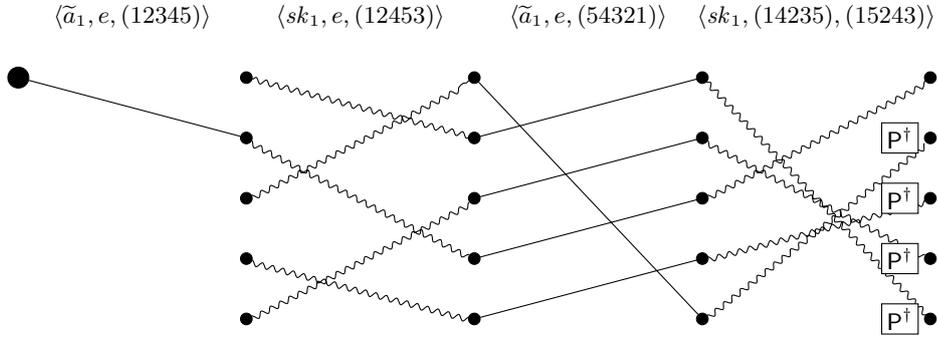
\begin{figure}
\centering
\begin{tikzpicture}[scale=0.8,every node/.style={draw,shape=circle,fill,scale=0.5},
decoration={snake, segment length=4, amplitude=1}]

\foreach \x in {2,...,5}
{
	\foreach \y in {1,...,5}
	{
		\pgfmathtruncatemacro{\label}{\x + 5 * \y}
		\node (\x-\y) at (3.75*\x,6-\y) {};
	}
}
\node[style={draw,shape=circle,fill,scale=1.8}] (1-1) at (3.75*1,6-1) {};

\draw     (1-1) -- (2-2);
          
\draw[decorate] (2-1) -- (3-2)
          (2-2) -- (3-4)
          (2-4) -- (3-5)
          (2-5) -- (3-3)
          (2-3) -- (3-1);

\draw     (3-5) -- (4-4)
          (3-4) -- (4-3)
          (3-3) -- (4-2)
          (3-2) -- (4-1)
          (3-1) -- (4-5);
 
 \draw[decorate]  (4-1) -- (5-5)
          (4-5) -- (5-2)
          (4-2) -- (5-4)
          (4-4) -- (5-3)
          (4-3) -- (5-1);

\filldraw[fill=white] (3.75*5 - 0.8,6-2 - 0.25) rectangle (3.75*5 - 0.2,6-2 + 0.25);
\node[draw=none,fill=none,scale=2] at (3.75*5 - 0.5, 6-2) {$\P^\dag$};
\filldraw[fill=white] (3.75*5 - 0.8,6-3 - 0.25) rectangle (3.75*5 - 0.2,6-3 + 0.25);
\node[draw=none,fill=none,scale=2] at (3.75*5 - 0.5, 6-3) {$\P^\dag$};
\filldraw[fill=white] (3.75*5 - 0.8,6-4 - 0.25) rectangle (3.75*5 - 0.2,6-4 + 0.25);
\node[draw=none,fill=none,scale=2] at (3.75*5 - 0.5, 6-4) {$\P^\dag$};
\filldraw[fill=white] (3.75*5 - 0.8,6-5 - 0.25) rectangle (3.75*5 - 0.2,6-5 + 0.25);
\node[draw=none,fill=none,scale=2] at (3.75*5 - 0.5, 6-5) {$\P^\dag$};

\node[draw=none,fill=none,scale=2] at (5.625,6) {$\langle \encrypted{a}_1, e, (12345)\rangle$};
\node[draw=none,fill=none,scale=2] at (9.375,6) {$\langle \sk_1, e, (12453)\rangle$};
\node[draw=none,fill=none,scale=2] at (13.125,6) {$\langle \encrypted{a}_1, e, (54321)\rangle$};
\node[draw=none,fill=none,scale=2] at (16.875,6) {$\langle \sk_1, (14235), (15243)\rangle$};

\end{tikzpicture}
\caption{Structure of the (first half of the) gadget, with measurements, coming from the 5-permutation branching program for the $\OR$ function on the input $(0,0)$. The example program's instructions are displayed above the permutations. The solid  lines correspond to Bell measurements, while the wavy lines
represent EPR pairs.}
\label{fig:5pbp}
\end{figure}

\paragraph{Example.}
The $\OR$ function on two bits can be computed using a width-5 permutation branching program of length 4, consisting of the following list of instructions:
\begin{enumerate}
\item $\langle 1, e, (12345)\rangle$
\item $\langle 2, e, (12453)\rangle$
\item $\langle 1, e, (54321)\rangle$
\item $\langle 2, (14235), (15243)\rangle$
\end{enumerate}
As a simplified example, suppose the decryption function $\HE.\Dec_{\sk}(\encrypted{a})$ is $\sk_1 \,\OR\,\, \encrypted{a}_1$. Then, for one possible example set of values of $\encrypted{a}$ and $\sk$, half of the gadget and measurements will be 
as given in Figure~\ref{fig:5pbp}. To complete this gadget, the same construction
is appended, reflected horizontally.

\subsection{For log-space computable decryption functions}\label{sec:gadget-construction-logspace}
Even though the construction based on Barrington's theorem has enough power for current classical homomorphic schemes, it is possible to improve on this construction
in two directions. Firstly, we extend our result to be able to handle a larger class of decryption functions: those that can be computed in logarithmic space, instead of only $\NC^1$. Secondly, for some specific decryption functions, executing the construction of Section~\ref{sec:barrington} might produce significantly larger 
gadgets than necessary. For instance, even for very simple circuits of depth $\log \kappa$, Barrington's theorem produces programs of length $\kappa^2$ ---
a direct approach can often easily improve on the exponent of the polynomial. See also the garden-hose protocols
for equality \cite{Mar14,CSWX14} and the majority function \cite{KP14} for examples of non-trivial protocols
that are much more efficient than applying Barrington's theorem
as a black box.

In Appendix~\ref{app:gadget-gh} we describe a construction for log-space computation in depth. 
The explanation in the appendix uses a different language than the direct encoding of the previous section:
there is a natural way of writing the requirements on the gadgets as a two-player task, and then writing strategies for this task in the \emph{garden-hose model}.

\begin{theorem}\label{thm:gadget-logspace}
Let $\HE.\Dec_{\sk}(\encrypted{a})$ be the decryption function of the classical homomorphic encryption scheme $\HE$. If $\HE.\Dec$ is computable by a Turing machine that uses space $O(\log \kappa)$, where $\kappa$ is the security parameter,
then there exist gadgets for $\SCHEME$ of size polynomial in $\kappa$.
\end{theorem}

Writing these gadgets in terms of the garden-hose model, even though it adds a layer of complexity to the construction, gives more insight into the structure of the gadgets, and forms its original inspiration. We therefore sketch the link between log-space computation and gadget construction within this framework.

Viewing the gadget construction as instance of the garden-hose model, besides clarifying the log-space construction,
also makes it easier to construct gadgets for specific cases.
Earlier work developed protocols in the garden-hose model for several functions, see for instance~\cite{Spe11, BFSS13, KP14}, and connections to other models of computation.
These results on the garden-hose model might serve as building blocks to create more efficient gadgets for specific decoding functions of classical homomorphic schemes, that are potentially much smaller than those created as a result of following the general constructions of Theorem~\ref{thm:gadget-logdepth} or~\ref{thm:gadget-logspace}.

\subsection{Specific case: Learning With Errors}
The scheme by Brakerski and Vaikuntanathan \cite{BV11} is well-suited
for our construction, and its decryption function is representative for a much wider class of schemes which are
based on the hardness of Learning With Errors (LWE).
As an example, we construct gadgets that enable quantum homomorphic encryption based
on the BV11 scheme. Let $\kappa$ be the security parameter, and let $p$ be the modulus 
of the integer ring over which the scheme operates.

The ciphertext $c$ is given by a pair $(\mathbf(v), w)$, with $\mathbf(v) \in \mathbb{Z}^\kappa_p$ and $w \in \mathbb{Z}_p$.
The secret key $\mathbf{s}$ is an element of $\mathbf{Z}^k_p$.
The decryption of a message $m$ involves computation of an inner product over the ring $\mathbb{Z}_p$,
\begin{equation}
m = (w - \langle \mathbf{v}, \mathbf{s}\rangle) \pmod p \pmod 2 \,. \label{eq:bv11decrypt}
\end{equation}

The BV11 scheme is able to make the modulus small, i.e.~polynomial in $\kappa$, before encryption.
In Appendix~\ref{app:bv11construction} we present an explicit construction for the case of small modulus $p$, which can be illustrative to read as example of our construction, and an implicit construction for more complicated gadgets for the case of superpolynomial $p$. 

\begin{proposition}
 The decryption function of the BV11 scheme translates into polynomial-sized gadgets.
\end{proposition}

\subsection{Constructing gadgets using limited quantum resources}
In a setting where a less powerful client wants to delegate some quantum computation to a more powerful server, it is important to minimize the amount of effort required from the client. In delegated quantum computation, the complexity of a protocol can be measured by, among other things, the total amount of communication between client and server, the number of rounds of communication, and the quantum resources available to the client, such as possible quantum operations and memory size.

We claim that $\SCHEME$ gives rise to a three-round delegated quantum computation protocol in a setting where the client can perform only Pauli and swap operations. $\SCHEME.\Enc$ and $\SCHEME.\Dec$ only require local application of Pauli operators to a quantum state, but $\SCHEME.\KeyGen$ is more involved because of the gadget construction. However, when supplied with a set of EPR pairs from the server (or any other untrusted source), the client can generate the quantum evaluation key for \SCHEME using only Pauli and swap operations. Even if the server produces some other state than the claimed list of EPR pairs, the client can prevent the leakage of information about her input by encrypting the input with random Pauli operations. More details are supplied in Appendix~\ref{app:gadget-construction-limited-quantum-power}.

Alternatively, $\SCHEME$ can be regarded as a two-round delegated quantum computation protocol in a setting where the client can perform arbitrary Clifford operations, but is limited to a constant-sized quantum memory, given that $\HE.\Dec$ is in $\NC^1$. In that case, the gadgets can be constructed ten qubits at a time, by constructing the sets of five EPR pairs as specified in Section~\ref{sec:barrington}. By decomposing the 5-cycles into products of 2-cycles, the quantum memory can even be reduced to only four qubits. The client sends these small parts of the gadgets to the server as they are completed. Because communication remains one-way until all gadgets have been sent, this can be regarded as a single round of communication.

\section{Conclusion}\label{sec:conclusion}
We have presented the first quantum homomorphic encryption scheme \SCHEME that is compact and allows evaluation of circuits with polynomially many \T gates in the security parameter, i.e.~arbitrary polynomial-sized circuits. Assuming that the number of wires involved in the evaluation circuit is also polynomially related to the security parameter, we may consider \SCHEME to be leveled fully homomorphic. The scheme is based on an arbitrary classical FHE scheme, and any computational assumptions needed for the classical scheme are also required for security of \SCHEME. However, since \SCHEME uses the classical FHE scheme as a black box, any FHE scheme can be plugged in to change the set of computational assumptions.

Our constructions are based on a new and interesting connection between the area of instantaneous non-local quantum computation and quantum homomorphic encryption. Recent techniques developed by Speelman~\cite{Spe15arxiv}, based on the garden-hose model~\cite{BFSS13}, have turned out to be crucial for our construction of quantum gadgets which allow homomorphic evaluation of \T gates on encrypted quantum data. 

\subsection{Future work} Since Yu, P\'erez-Delgado and Fitzsimons~\cite{YPF14} showed that
information-theoretically secure QFHE is impossible (at least in the exact case), it is natural to wonder whether it is possible to construct a non-leveled QFHE scheme based on computational assumptions. If such a scheme is not possible, can one find lower bounds on the size of the evaluation key of a compact scheme? Other than the development of more efficient QFHE schemes, one can consider the construction of QFHE schemes with extra properties, such as circuit privacy against active adversaries. It is also interesting to look at other cryptographic tasks that might be executed using QFHE. In the classical world for example, multiparty computation protocols can be constructed from fully homomorphic encryption~\cite{CDN01}.
We consider it likely that our new techniques will also be useful in other contexts such as quantum indistinguishability obfuscation~\cite{AF16arxiv}.

\ifbool{anonymous}{}{
\section*{Acknowledgements}
We acknowledge useful discussions with Anne Broadbent, Harry Buhrman, and Leo Ducas. We thank Stacey Jeffery for providing the inspiration for a crucial step in the security proof, and Gorjan Alagic and Anne Broadbent for helpful comments on a draft of this article. This work was supported by the 7th framework EU SIQS and QALGO, and a NWO VIDI grant. 
}

\bibliographystyle{alpha}
\bibliography{homomorphic}

\newpage
\appendix

\section{Key update rules}\label{app:key-update}
\subsection{Applying Clifford group gates}\label{app:key-update-clifford}
For convenience, we repeat the key update rules when applying the generators of the Clifford group to a state
that is encrypted with the quantum one-time pad. These can be found in many places in the literature (or can be easily calculated by hand), see also e.g.~\cite[Appendix C]{BJ15}.

After applying a gate to the $i$th wire of a quantum state that has one-time pad keys $a_i$ and $b_i$, we update the keys as
\[
\P_i : (a_i, b_i) \to (a_i, a_i \oplus b_i)
\]
and
\[
\H_i : (a_i, b_i) \to (b_i, a_i) \,.
\]

For the two-qubit $\cnot$ gate applied on control wire $i$, with target $j$, we update the corresponding keys as
\[
\cnot_{i,j} : (a_i, b_i; a_j, b_j) \to (a_i, b_i \oplus b_j; a_i \oplus a_j, b_j)  \,.
\]

\subsection{Using the gadget}\label{app:key-update-gadget}
After using the gadget, but before updating any classical information, the evaluator has: 
the encrypted one-time pad keys $\encrypted{a}$, $\encrypted{b}$, a list of $m$ pairs for Bell measurements $M \leftarrow \SCHEME.\GenMeasurement(\encrypted{a})$
and a list of outcomes for each of these $m$ measurements, say $c,d \in \{0,1\}^m$. 

The evaluator also has encrypted versions of $\{ (s_1, t_1), (s_2, t_2), \dots, (s_m, t_m) \}$, $p \in \{0,1\}^m$, and $x,z \in \{0,1\}^m$ that specify the structure of the gadget.

Say an arbitrary qubit\footnote{The input qubit is not necessarily a pure state, but we write an arbitrary pure
state without loss of generality, to simplify notation.} was teleported through the gadget, so that the qubit started in some state $\P^a \X^a \Z^b \ket{\psi}$ and is currently in state $\X^{a'} \Z^{b'} \ket{\psi}$. 
We sketch the algorithm an evaluator would execute on this encrypted state, to compute (encrypted versions of)
the updated keys $a'$ and $b'$.
Updating the keys is not complicated, it mostly involves bookkeeping to keep track
of the current location of the qubit, and its current $\X$-correction, $\Z$-correction and phase.

We explain the calculation as if performed with
the unencrypted versions; in the actual execution, only the encrypted versions of all variables are used, and this entire calculation
is performed homomorphically. Since all the mentioned classical information either is or can be encrypted with the same public key, this calculation
can be handled by the classical homomorphic scheme $\HE$.

The algorithm tracks the path the qubit takes through the gadget, by resolving the teleportations that involve the qubit one by one.
Even though the measurements were all performed at the same time, we will describe them as if ordered in this manner.
All additions of the keys of the one-time pad will be performed modulo 2, since $\X^2=\Z^2=\id$.

Let $\mathbf{a}, \mathbf{b}$ be variables that hold the current key to the one-time pad at every step of the algorithm. We initialize these as $\mathbf{a} \leftarrow a$ and $\mathbf{b} \leftarrow b$.
Let $\mathbf{q}$ be a variable that stores whether or not the qubit currently has an extra phase gate, initialized as $\mathbf{q} \leftarrow a$.
Let $\mathbf{r}$ be the variable that contains the current location of the qubit, with possible locations 0 to $2m$, initialized to 0. 
That is, we view the current state as being $\P^\mathbf{q} \X^{\mathbf{a}} \Z^{\mathbf{b}} \ket{\psi}$
at location $\mathbf{r}$.
For every step we update the location depending on $M$ and $\{ (s_1, t_1), (s_2, t_2), \dots, (s_m, t_m) \} $, and update the keys depending on the corresponding measurement outcomes.

First, find the pair in $M$ that contains the current location $\mathbf{r}$, say pair $i$ which consists of $(r,s)$
for some other location $s$. The outcome of this measurement is given by $c[i]$ and $d[i]$.
Effectively, these outcomes change the current state to 
\[
\X^{c[i]} \Z^{d[i]} \P^{\mathbf{q}} \X^{\mathbf{a}} \Z^{\mathbf{b}} \ket{\psi} = \P^{\mathbf{q}} \X^{\mathbf{a} + c[i]} \Z^{\mathbf{b} + d[i] + \mathbf{q} \cdot c[i]} \ket{\psi} \,,
\]
therefore we update $\mathbf{a} \leftarrow \mathbf{a} + c[i]$ and $\mathbf{b} \leftarrow \mathbf{b} + d[i] + q \cdot c[i]$. Note that the key update
rules also involve multiplication -- an extra $\Z$ gate is added if the phase gate was present, $\mathbf{q}=1$, and the teleportation measurement required an X correction, $c[i]=1$.

Next, find the pair in $\{ (s_1, t_1), (s_2, t_2), \dots, (s_m, t_m) \} $ that contains the new location $s$, say pair $j$ containing $(s, t)$. The teleportation
of the qubit through this pair effectively applies $\X^{x[j]} \Z^{z[j]} (\P^{\dagger})^{p[j]}$ to the state.
Then, if we already use the updated $\mathbf{a}$ and $\mathbf{b}$, the quantum state at this step equals
\begin{align*}
 \X^{x[j]} \Z^{z[j]} (\P^{\dagger})^{p[j]}  \P^{\mathbf{q}} \X^{\mathbf{a}} \Z^{\mathbf{b}} \ket{\psi} &= \X^{x[j]} \Z^{z[j]} \P^{p[j]}  \P^{\mathbf{q}} \X^{\mathbf{a}} \Z^{\mathbf{b} + p[j]} \ket{\psi} \\
 &= \P^{p[j] + \mathbf{q} \!\!\! \pmod 2} \X^{\mathbf{a} + x[j]} \Z^{\mathbf{b} + z[j] + p[j] \cdot (1 + \mathbf{q}) +  x[j]\cdot (p[j] + \mathbf{q})} \ket{\psi}\,. \\
\end{align*}
For rewriting, we used the fact that $\P^2 = \Z$ and that $\P^\dag = \P \Z$, together with the commutation relations
from the previous section. We therefore update the phase $\mathbf{q} \leftarrow p[j] + \mathbf{q}  \pmod 2$,
and the components of the quantum one-time pad to $\mathbf{a} \leftarrow \mathbf{a} + x[j]$ and $\mathbf{b} \leftarrow \mathbf{b} + z[j] + p[j] \cdot (1 + q[j]) +  x[j]\cdot (p[j] + q[j])$.
Finally, set the new location of the qubit $\mathbf{r} \leftarrow t$.

The previous two steps are then repeated $m$ times, where $2m$ is the size of the gadget, to eventually (homomorphically) compute the new updated keys $a'$,$b'$ to
the quantum one-time pad. Afterwards, all temporary variables can be discarded, and only the updated keys will be needed
for continuining the protocol.

\section{Circuit privacy}\label{app:circuit-privacy}

In this appendix, we demonstrate that with only a slight modification of $\SCHEME$, the scheme has circuit privacy in the semi-honest setting, i.e. against passive adversaries. Classically, circuit privacy is defined by requiring the existence of a simulator $\Sim_\HE$ whose inputs are the public parameters and $\C(x)$ and which produces an  output which is indistinguishable from the homomorphic evaluation of $\C$ on the encryption of $x$. Formally, circuit privacy is defined as follows.

\begin{definition}[Classical circuit privacy -- semi-honest setting \cite{IP07}]\label{def:cp-classical-honest}\\
A classical homomorphic encryption scheme $\HE$ has statistical circuit privacy in the semi-honest (`honest-but-curious') model if there exists a $PPT$ algorithm $\Sim_\HE$ and a negligible function $\eta$ such that for any security parameter $\kappa$, input $x$, key set $(\pk, \evk, \sk) \leftarrow \HE.\KeyGen(1^\kappa)$, and circuit $\C$:
\[
\delta (\HE.\Eval^{\C}_{\evk}(\HE.\Enc_{pk}(x)), \Sim_\HE(1^\kappa, \pk, \evk, \C(x))) \leq \eta(\kappa)
\]
\end{definition}
Here, $\delta(X, Y) := \frac{1}{2}\sum_{u \in U} \bigl\lvert \Pr[X = u] - \Pr[Y = u] \bigr\rvert$ is the \emph{statistical distance} between two random variables over a finite universe $U$. For notational convenience, we will often write $\Sim_\HE(\C(x))$ if the rest of the arguments are clear from the context. Also we will sometimes write $\X \approx_{a} Y$ to denote that $\delta(X,Y) \leq a$.

If the recryption functionality $\Rec_{i \to j}$ is defined as the composition of the procedures $\HE.\Eval_{\evk_j}^{\HE.\Dec_i}$ and $\HE.\Enc_{\pk_j}$, as in Section~\ref{sec:he-definition}, then recryptions do not degrade the privacy of the computation: a homomorphic evaluation of some function with key switching is statistically close to running the simulator directly on the function output using only the \emph{last} key set.

\begin{lemma}\label{lem:circuit-privacy-keyswitch}
Suppose $\HE$ has statistical circuit privacy in the semi-honest setting, and let $\Sim_{\HE}$ and $\eta$ be as in Definition~\ref{def:cp-classical-honest}. Then for any security parameter $\kappa$, $L$ polynomial in $\kappa$, list of circuits $\C_1, ..., \C_L$ and list of keysets $(\pk_i, \evk_i, \sk_i)_{i=1}^L$ generated by $\HE.\KeyGen(1^\kappa)$, and input $x$, the statistical distance between
\[
\HE.\Eval_{\evk_L}^{\C_L}(\HE.\Rec_{(L-1) \to L}(\HE.\Eval_{\evk_{L-1}}^{\C_{L-1}}(\cdots \HE.\Eval_{\evk_1}^{\C_1}(\HE.\Enc_{\pk_1}(x))))
\]
and
\[
\Sim_{\HE}(1^{\kappa},\pk_L,\evk_L,\C_L(\C_{L-1}(\cdots C_1(x))))
\]
is negligible in $\kappa$.
\end{lemma}

\begin{proof}
Since $\HE.\Rec_{(L-1) \to L} = \HE.\Eval_{\evk_L}^{\HE.\Dec_{\sk_{L-1}}} \circ \HE.\Enc_{\pk_L}$ by definition, we have that
\begin{align*}
& &&\HE.\Eval_{\evk_L}^{\C_L}(\HE.\Rec_{(L-1) \to L}(\HE.\Eval_{\evk_{L-1}}^{\C_{L-1}}(\cdots \HE.\Eval_{\evk_1}^{\C_1}(\HE.\Enc_{\pk_1}(x))))\\
&= &&\HE.\Eval_{\evk_L}^{\C_L \circ \HE.\Dec_{\sk_{L-1}}}(\HE.\Enc_{\pk_L}(\HE.\Eval_{\evk_{L-1}}^{\C_{L-1}}(\cdots \HE.\Eval_{\evk_1}^{\C_1}(\HE.\Enc_{\pk_1}(x))))\\
&\approx_{\eta(\kappa)} &&\Sim_{\HE}(1^\kappa,\pk_L,\evk_L,\C_L(\HE.\Dec_{\sk_{L-1}}(\HE.\Eval_{\evk_{L-1}}^{\C_{L-1}}(\cdots \HE.\Eval_{\evk_1}^{\C_1}(\HE.\Enc_{\pk_1}(x))))))
\end{align*}
which, by correctness of $\HE$, is statistically indistinguishable from
\[
\Sim_{\HE}(1^{\kappa},\pk_L, \evk_L, \C_L(\C_{L-1}(\C_{L-2}(\cdots \C_1(x)))))
\]
as long as $L$ is polynomial in $\kappa$. By triangle inequality, the statement of the lemma follows.
\end{proof}

In the quantum setting, we need to take into account the fact that the input state may be part of some larger (possibly entangled) system. This leads to the following definition of \emph{quantum circuit privacy} in the semi-honest setting:

\begin{definition}[Quantum circuit privacy -- semi-honest setting]\label{def:circuit-privacy}
A quantum homomorphic encryption scheme $\QHE$ has statistical circuit privacy in the semi-honest setting if there exists a quantum $PPT$ algorithm $\Sim_{\QHE}$ and a negligible function $\eta$ such that for any security parameter $\kappa$, depth parameter $L$, key set $(\pk, \rho_{\evk},\sk) \leftarrow \QHE.\KeyGen(1^\kappa,1^L)$, state $\sigma$, and circuit $\C$ with up to $L$ $\T$-gates:
\[
\left\|
\left(\Phi_{\QHE.\Eval}^{\C,\rho_{\evk},\pk} \circ \Phi_{\QHE.\Enc}^{\pk}\right)
-
\left(\Phi_{\Sim_{\QHE}}^{\rho_{\evk},\pk} \circ \Phi_{\C}\right)
\right\|_{\lozenge} \leq \eta(\kappa)
\]
\end{definition}
In this definition, $\Phi_{\U}$ denotes the quantum channel induced by
the circuit or functionality $\U$. The diamond norm
$\left\|\Phi_U\right\|_{\lozenge}$ is defined in terms of the trace
norm: $\left\|\Phi_U\right\|_{\lozenge} := \max_{\sigma} \left\|(\Phi_U \otimes \id) \sigma\right\|_1$ where the maximisation is over input states $\sigma$.

We now show that the scheme $\SCHEME$ can, with very little overhead, be modified to provide circuit privacy, as stated in Theorem~\ref{thm:circuit-privacy} from Section~\ref{sec:circuit-privacy}:

\circuitprivacy*

\begin{proof}
We make the following alteration to the scheme $\SCHEME$: at the end of the evaluation procedure, the evaluator applies a (random) quantum one-time pad to the output of the evaluation, and updates the classical encryptions of the keys accordingly. The rest of the scheme remains exactly the same, and it is clear that this altered version of $\SCHEME$ is still compact and correct.

Intuitively, the randomization step at the end of the evaluation phase completely hides the circuit: the keys to the quantum one-time pads themselves are now entirely independent of the circuit, and circuit privacy of $\HE$ will ensure that even the classical encryption of these keys does not reveal any information about the computations performed on them.

To formalize this intuition, we define a quantum algorithm $\Sim_{\SCHEME}$ satisfying the constraints given in Definition~\ref{def:circuit-privacy}. Let $\Sim_{\HE}$ be the classical simulator guaranteed to exist by the classical circuit privacy of $\HE$ (see Definition~\ref{def:cp-classical-honest}). Given some security parameter $\kappa$, some keys $\pk = (\pk_1, ..., \pk_L)$ and $\evk = (\evk_1, ..., \evk_L)$, and some quantum state $\sigma$, let $\Sim_{\SCHEME}$ apply a uniformly random quantum one-time pad to $\sigma$, and apply $\Sim_{\HE}(1^{\kappa}, \pk_L, \evk_L, \cdot)$ to the pad keys. The resulting classical-quantum state is the output of $\Sim_{\SCHEME}$. This algorithm resembles $\SCHEME.\Enc$, but instead of calling $\HE.\Enc$ (with $\pk_1$) as a subroutine, it handles the pad key information using the classical simulator $\Sim_{\HE}$ (with $\pk_L$).

If we can show that the trace distance
\[
\left\| \left(\left(\Phi_{\SCHEME.\Eval}^{\C,\rho_{\evk}, \pk} \circ \Phi_{\SCHEME.\Enc}^{\pk} \right) \otimes \id\right)\sigma - \left((\Phi_{\Sim_{\SCHEME}}^{\rho_\evk,\pk} \circ \Phi_{\C}) \otimes \id\right) \sigma\right\|_1
\]
is negligible for any quantum state $\sigma$ of an appropriate dimension, then quantum circuit privacy of $\SCHEME$ immediately follows from Defintion~\ref{def:circuit-privacy} and the definition of the diamond norm.

Write $\sigma_{sim} := \left((\Phi_{\Sim_{\SCHEME}}^{\rho_\evk,\pk} \circ \Phi_{\C}) \otimes \id\right) \sigma$, and $\sigma_{eval} := \left(\left(\Phi_{\SCHEME.\Eval}^{\C,\rho_{\evk}, \pk} \circ \Phi_{\SCHEME.\Enc}^{\pk} \right) \otimes \id\right)\sigma$. We study the state $\sigma_{sim}$ in more detail, and show how to transform it into $\sigma_{eval}$ in only a few (negligible) steps. As a result, the trace distance of these two states will be negligible.

By definition of the algorithm $\Sim_{\SCHEME}$, the state $\sigma_{sim}$ is equal to
\begin{multline}
\frac{1}{2^{2n}} \enskip \sum_{\mathclap{x,z \in \{0,1\}^n}} \enskip \,\, \bigg(\bigotimes_{i = 1}^n \rho\big(\Sim_\HE(1^{\kappa}, \pk_L, \evk_L, x[i])\big) \otimes \bigotimes_{i = 1}^n \rho\big(\Sim_\HE(1^{\kappa}, \pk_L, \evk_L, z[i])\big) \otimes\\
\left(\left(\bigotimes_{i=1}^n \X^{x[i]}\Z^{z[i]} \C \otimes \id\right) \sigma \left(\C^{\dag} \bigotimes_{i=1}^n \X^{x[i]}\Z^{z[i]} \otimes \id \right)\right)\bigg). \nonumber
\end{multline}
During the evaluation procedure of $\SCHEME$, the evaluator updates the keys to the quantum one-time pad for all $n$ qubits in the circuit. These updates depend on the circuit that is being evaluated, some randomness $r$ from the Bell measurement outcomes\footnote{Although for the scheme $\SCHEME$, the measurement outcomes will in principle be uniformly distributed, we will not make this assumption here. In case of a malicious key generator, measurement outcomes might be correlated in some way. Therefore, we will simply assume that $r$ is distributed according to some distribution $R$.} and of course on the initial one-time pad keys. Let $f_i^{\C,r}(a,b)$ denote the $\X$ key on the $i$th qubit after the evaluation of some circuit $\C$ with randomness $r$, with $a,b \in \{0,1\}^n$ the initial pad keys before the evaluation procedure. Similarly, let $g_i^{\C,r}(a,b)$ denote the $\Z$ key for that qubit.

At the end of the evaluation phase, the evaluator chooses bit strings $x$ and $z$ uniformly at random, so the final keys $f_i^{\C,r}(a,b) \oplus x[i]$ and  $g_i^{\C,r}(a,b) \oplus z[i]$ are themselves completely uniform for any $a,b$. Therefore, the state $\sigma_{sim}$ is actually equal to

\begin{outdent}
\begin{align*}
\frac{1}{2^{4n}} \enskip\enskip\enskip \sum_{\mathclap{\substack{a,b,x,z \in \{0,1\}^n \\ r \in \{0,1\}^*} } } \enskip\enskip\enskip\enskip \Pr_R(r) \bigg(&\bigotimes_{i = 1}^n \rho\big(\Sim_\HE(1^{\kappa}, \pk_L, \evk_L, f_i^{\C,r}(a,b) \oplus x[i])\big) \otimes\\
&\bigotimes_{i = 1}^n \rho\big(\Sim_\HE(1^{\kappa}, \pk_L, \evk_L, g_i^{\C,r}(a,b)\oplus z[i])\big) \otimes\\
\noalign{$\displaystyle\left(\left(\bigotimes_{i=1}^n \X^{f_i^{\C,r}(a,b) \oplus x[i]}\Z^{g_i^{\C,r}(a,b) \oplus z[i]} \C \otimes \id\right) \sigma \left(\C^{\dag} \bigotimes_{i=1}^n \X^{f_i^{\C,r}(a,b) \oplus x[i]}\Z^{g_i^{\C,r}(a,b) \oplus z[i]} \otimes \id \right)\right)
\bigg).$}
\end{align*}
\end{outdent}
This is where the classical circuit privacy property kicks in: for any fixed $i,a,b,\C,r,x$, the result of the probabilistic computation $\Sim_{\HE}(f_i^{\C,r}(a,b) \oplus x[i])$ is statistically indistinguishable from the evaluation of the function $f_i^{\C,r}(\cdot,\cdot)\oplus x[i]$ on the encryptions of $a$ and $b$. Note however that the evaluation of $f_i^{\C,r}$ is performed in several steps, with key switching in between. That is, separate functions $h_1$ through $h_L$ are evaluated in each key set $1$ through $L$, such that $f_i^{\C,r} = h_L \circ \cdots \circ h_1$. We abstract away from the exact way that the function $f_i^{\C,r}$ is broken up into these separate functions $h_1, ..., h_L$, and simply write $\HE.\Eval^{f_i^{\C,r}(\cdot,\cdot)\oplus x[i]}_{1,...,L}(\HE.\Enc_{\pk_1}(a,b))$ to denote
\[
\HE.\Eval_{\evk_L}^{(\cdot \oplus x[i]) \circ h_L}(\HE.\Rec_{(L-1) \to L}(\HE.\Eval_{\evk_{L-1}}^{h_{L-1}}(\cdots \HE.\Eval_{\evk_1}^{h_1}(\HE.\Enc_{\pk_1}(a,b)))).
\]
By Lemma~\ref{lem:circuit-privacy-keyswitch}, it follows that
\begin{outdent}
\[
\delta \Big(\HE.\Eval^{f^{\C,r}_i(\cdot,\cdot) \oplus x[i]}_{1,...,L}(\HE.\Enc_{\pk_1}(a,b)), \Sim_\HE(1^\kappa, \pk_L, \evk_L, f_i^{\C,r}(a,b) \oplus x[i])\Big) \leq \eta(\kappa)
\]
\end{outdent}
for some negligible function $\eta$. We can rewrite this equation in terms of the trace distance to get
\begin{outdent}
\[
\left\| \Big(\HE.\Eval^{f^{\C,r}_i(\cdot,\cdot) \oplus x[i]}_{1,...,L}(\HE.\Enc_{\pk_1}(a,b)) - \Sim_\HE(1^\kappa, \pk_L, \evk_L, f_i^{\C,r}(a,b) \oplus x[i])\right\|_1 \leq 2\eta(\kappa)
\]
\end{outdent}
A similar result holds for $g_i^{\C,r}(\cdot,\cdot)\oplus z[i]$. Using subadditivity of the trace norm with respect to the tensor product, it follows that the trace distance between $\sigma_{sim}$ and
\begin{outdent}
\begin{align*}
\frac{1}{2^{4n}} \enskip\enskip\enskip \sum_{\mathclap{\substack{a,b,x,z \in \{0,1\}^n \\ r \in \{0,1\}^*} } } \enskip\enskip\enskip\enskip \Pr_R(r) \bigg(&\bigotimes_{i = 1}^n \rho\big(\HE.\Eval^{f^{\C,r}_i \oplus x[i]}_{1,...,L}(\HE.\Enc_{\pk_1}(a,b))\big) \otimes\\
&\bigotimes_{i = 1}^n \rho\big(\HE.\Eval^{g^{\C,r}_i \oplus z[i]}_{1,...,L}(\HE.\Enc_{\pk_1}(a,b))\big) \otimes\\
\noalign{$\displaystyle\left(\left(\bigotimes_{i=1}^n \X^{f_i^{\C,r}(a,b) \oplus x[i]}\Z^{g_i^{\C,r}(a,b) \oplus z[i]} \C \otimes \id\right) \sigma \left(\C^{\dag} \bigotimes_{i=1}^n \X^{f_i^{\C,r}(a,b) \oplus x[i]}\Z^{g_i^{\C,r}(a,b) \oplus z[i]} \otimes \id \right)\right)\bigg)$}
\end{align*}
\end{outdent}
is at most $4n \cdot \eta(\kappa)$. Note that this last state is exactly $\sigma_{eval}$, the result of putting $\sigma$ through the channel ($\Phi_{\SCHEME.\Eval}^{\C,\rho_{\evk},\pk} \circ \Phi_{\SCHEME.\Enc}^{\pk}) \otimes \id$. We conclude that for any $\sigma$,
\[
\|\sigma_{eval} - \sigma_{sim}\|_1 \leq 4n \cdot \eta(\kappa)
\]
for some negligible function $\eta$ that does not depend on $\sigma$. Hence,
\begin{align*}
\left\|(\Phi_{\SCHEME.\Eval}^{\C,\rho_{\evk}, \pk} \circ \Phi_{\SCHEME.\Enc}^{\pk}) - (\Phi_{\Sim_{\SCHEME}}^{\rho_{\evk},\pk} \circ \Phi_{\C}) \right\|_{\lozenge} &=\\ 
\max_{\sigma} \left\| \left((\Phi_{\SCHEME.\Eval}^{\C,\rho_{\evk}, \pk} \circ \Phi_{\SCHEME.\Enc}^{\pk}) \otimes \id\right)\sigma - \left((\Phi_{\Sim_{\SCHEME}}^{\rho_{\evk},\pk} \circ \Phi_{\C}) \otimes \id\right) \sigma \right\|_{1} &\leq 4n \cdot \eta(\kappa)
\end{align*}
which is negligible if $\eta$ is negligible.
\end{proof}

\section{Gadget construction using the garden-hose model}\label{app:gadget-gh}

To construct the gadgets for specific decryption functions $\HE.\Dec$, we will consider a purified version of the construction of the gadget state.

Consider the following task among two parties Alice and Bob. Alice corresponds to the party which creates the gadget, so she has knowledge of the secret key $\sk$. Bob corresponds to the party which uses the gadget,
therefore Bob has some input $\encrypted{a}$ and a state $\P^{a} \ket{\psi}$,
where $a = \HE.\Dec_{\sk}(\encrypted{a})$. The end goal of the task is for Bob to possess the state $\X^{a'} \Z^{b'} \ket{\psi}$ for some $a'$,$b'$ that are computable
from classical information known to Alice and Bob. The players pre-share a number of EPR pairs between them,
and are only allowed to perform their actions without receiving any communication from the other player.
We only consider strategies where the players perform Bell measurements and inverse phase gates on
their local halves of the given EPR pairs (and in Bob's case also on the input qubit).

Before presenting strategies for this task, we first describe how this task translates to the creation of a gadget.
Say Alice and Bob share $2m$ EPR pairs between them (i.e.~they start with~$4m$ qubits in total).
Alice will perform Bell measurements between the halves of EPR pairs on her side, where the choices she makes depend on $\sk$. Since both players act before receiving any information from the other player, the actions of Alice and Bob are not
ordered -- we first consider how to describe the state when Alice has acted on her local half.
If Alice measures between, say, qubits $s$ and $t$,
with a two-bit outcome that describes the $\X$ and $\Z$ corrections, then we can instantly describe the qubits $s$ and $t$ on Bob's side as forming a fully entangled state --
this teleportation of EPR halves is sometimes called \emph{entanglement swapping}. Which out of the four Bell states is formed depends on the outcomes of Alice's measurement. 

We also allow Alice to perform a $\P^{\dagger}$ gate on some qubits before teleportations. Note that if Alice measures on the qubits given by $\{ (s_1, t_1), (s_2, t_2), \dots, (s_m, t_m) \}$, after applying an inverse phase gate when specified by the bit-string $p$, the state on Bob's side will exactly have the form
of $\gamma_{x,z}\bigl(g(\sk)\bigr)$, for some random binary strings $x,z$ that correspond to the outcomes of Alice's Bell measurements.
The quantum part of the gadget will be given by the reduced state on Bob's side,
while the measurement choices and outcomes of Alice will form the accompanying classical information.
The pairs that Bob chooses to perform Bell measurements on, which only depend on the encrypted information,
are exactly the output of the function $\SCHEME.\GenMeasurement$.

An upper bound to the hardness of this task is given by the \emph{garden-hose complexity} $\HE.\Dec$, written $\gh(\HE.\Dec)$, which is the least number of pipes needed for the players to compute it in the garden-hose model described in Section~\ref{sec:prelim-gardenhose}. This complexity measure is the main measure of hardness in the garden-hose model, and is relevant for the size of the gadgets in our construction.

The amount of space a Turing machine needs to compute any function $f$ is closely related to it garden-hose complexity $\gh(f)$.
The following theorem, proven in \cite{BFSS13}, provides us with a general way of transforming space-efficient algorithms into garden-hose protocols.

\begin{theorem}\cite[Theorem 2.12]{BFSS13}\label{thm:logspace}
If $f : \{0, 1\}^n \times \{0, 1\}^n \to \{0, 1\}$ is log-space computable, then $\gh(f)$ is
polynomial in $n$.
\end{theorem}

Since the garden-hose complexity is defined in a non-uniform way, the strategies of the players are not necessarily easily computable.
However, by inspection of the original proof, we see that the players effectively have to list all configurations for the Turing machine for $f$, and connect them according to the machine's transition function.
For a log-space decryption function $\HE.\Dec$, a player therefore only has to perform a polynomial-time computation to determine the strategy for a specific input.

The general construction is a direct consequence of the following lemma\footnote{The names of Alice and Bob have been swapped in order to fit the framework of this paper.}, which was recently derived in the context of instantaneous non-local quantum computation.

\begin{lemma}\label{lem:nonlocalphase} \cite[Lemma 8, paraphrased]{Spe15arxiv}
Assume Bob has a single qubit with state $\P^{f(x,y)} \ket{\psi}$, for
binary strings $x,y\in \{0,1\}^n$, where Alice knows the string $x$ and Bob knows $y$.
Let $\gh(f)$ be the garden-hose complexity of the function $f$. Then the following holds: 
 
\begin{enumerate}
\item There exists an instantaneous protocol without any communication which uses
$2 \gh(f)$ pre-shared EPR pairs after which a
known qubit of Bob is in the state
$\X^{g(\hat{x},\hat{y})} \Y^{h(\hat{x},\hat{y})} \ket{\psi}$. Here
$\hat{x}$ depends only on $x$ and the
measurement outcomes of Alice, and $\hat{y}$ depends on $y$ and
the measurement outcomes of Bob.
\item The garden-hose complexities of the functions $g$ and $h$ are at most linear in the complexity of the function $f$. 
\end{enumerate}
\end{lemma}

For our purposes, only the first part of this result is required. The construction used in this lemma is a direct
application of the garden-hose model~\cite{BFSS13},
together with a simplifying step which was inspired by results on the garden-hose model by Klauck and Podder~\cite{KP14}. 
In our case, the function $f$ will be the decryption function $\HE.\Dec$ where Alice holds $\sk$, and Bob holds $\encrypted{a}$.

\subsection{Toy example}\label{sec:toyexample}
The garden-hose protocols that correspond to actual decryption functions will quickly become complicated in their description.
Therefore, as an illustrative example, we will explicitly show how to convert a garden-hose protocol for the decryption function of a toy classical scheme $\TOY$ to a gadget.
We do not claim $\TOY$ to be homomorphic at all; we only define its very simple decryption function and leave the rest of the scheme undefined.

Consider the following definition of $\TOY.\Dec$ on ciphertext $c$ and key $\sk$ of a single bit:
\[
\TOY.\Dec_\sk(c) = \sk \oplus c \, .
\]
In Figure~\ref{fig:gadget-construction-1}, a garden-hose protocol~\cite{BFSS13} for $\TOY.\Dec$ is shown. For the protocol, Alice and Bob share three EPR-pairs which they use to teleport some qubit through, in a way that depends only on their own inputs $c$ (for Bob) and $\sk$ (for Alice). The qubit always starts in the location marked `in'. After the execution of the protocol, the qubit $\ket\psi$ should end up on Bob's side whenever $\TOY.\Dec_\sk(c) = 0$, and on Alice's side otherwise. For this small function, correctness is easily verified to hold for all possible inputs.

The \emph{garden-hose complexity} $\gh(\TOY.\Dec)$ is the minimum amount of EPR-pairs needed for the computation of $\TOY.\Dec$ in this way. See also~\cite{BFSS13}.

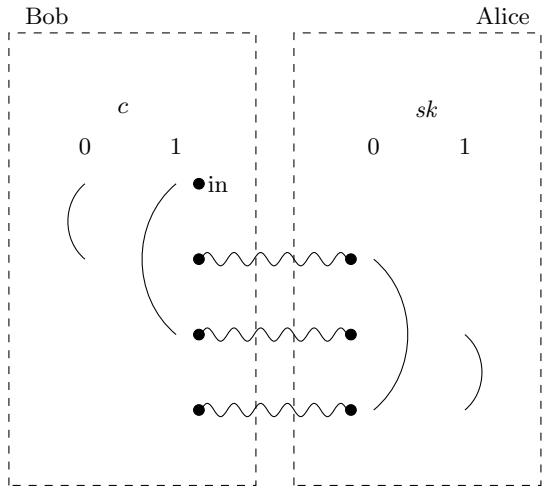
\begin{figure}[h!]
\centering
\begin{tikzpicture}[decoration=snake]
\filldraw (0,0) circle (2pt);
\filldraw (2,0) circle (2pt);
\draw[decorate] (0,0) -- (2,0);
\filldraw (0,1) circle (2pt);
\filldraw (2,1) circle (2pt);
\draw[decorate] (0,1) -- (2,1);
\filldraw (0,2) circle (2pt);
\filldraw (2,2) circle (2pt);
\draw[decorate] (0,2) -- (2,2);

\filldraw (0,3) circle (2pt);
\node[anchor=west] at (0,3) {in};

\node at (-1,4) {$c$};
\node at (-1.5,3.5) {$0$};
\node at (-0.3,3.5) {$1$};
\draw (-0.3,1) to [bend left=50] (-0.3,3);
\draw (-1.5,2) to [bend left=50] (-1.5,3);

\node at (3,4) {$\sk$};
\node at (2.3,3.5) {$0$};
\node at (3.5,3.5) {$1$};
\draw (2.3,0) to [bend right=50] (2.3,2);
\draw (3.5,0) to [bend right=50] (3.5,1);

\draw[dashed] (-2.5,-1) rectangle (0.75,5);
\draw[dashed] (1.25,-1) rectangle (4.5,5);
\node[anchor=south] at (-2,5) {Bob};
\node[anchor=south] at (4,5) {Alice};

\end{tikzpicture}
\caption{Garden-hose protocol for $\TOY.\Dec$. The snaky lines represent the EPR-pairs that form the resources to the protocol. The bended lines represent Bell measurements that Bob and Alice perform dependent on their inputs. For example, if $c = 0$ and $\sk = 0$, the qubit starting at `in' is teleported through the first EPR-pair by Bob, then back through the third EPR-pair by Alice. It comes out on the bottom location on Bob's side.}
\label{fig:gadget-construction-1}
\end{figure}

Suppose that Bob teleports some qubit $\P^a\X^a\Z^b \ket{\psi}$ through the protocol, and sets his input $c$ to be $\encrypted{a}$. Then whenever $a = \TOY.\Dec_{\sk}(\encrypted{a}) = 1$, the qubit will come out on Alice's side, and we will want to apply the correction $\P^\dag$. To make sure that the correction is applied to $\P^a\X^a\Z^b\ket\psi$, Alice can apply a $\P^\dag$ gate on all possible locations of the qubit. However, after this step, Alice and Bob do not know the location of the qubit (unless they share their inputs with one another non-homomorphically). The construction from~\cite[Lemma 8]{Spe15arxiv} solves this problem by applying the entire garden-hose protocol again in reverse: every EPR-half on which no measurement is performed, is connected through measurement with the EPR-half at the same position in the second copy of the protocol. That way, the (corrected) qubit follows the same path backwards, and always ends up on Bob's side at the `in' position of the second protocol (marked `out' in Figure~\ref{fig:gadget-construction-2}).

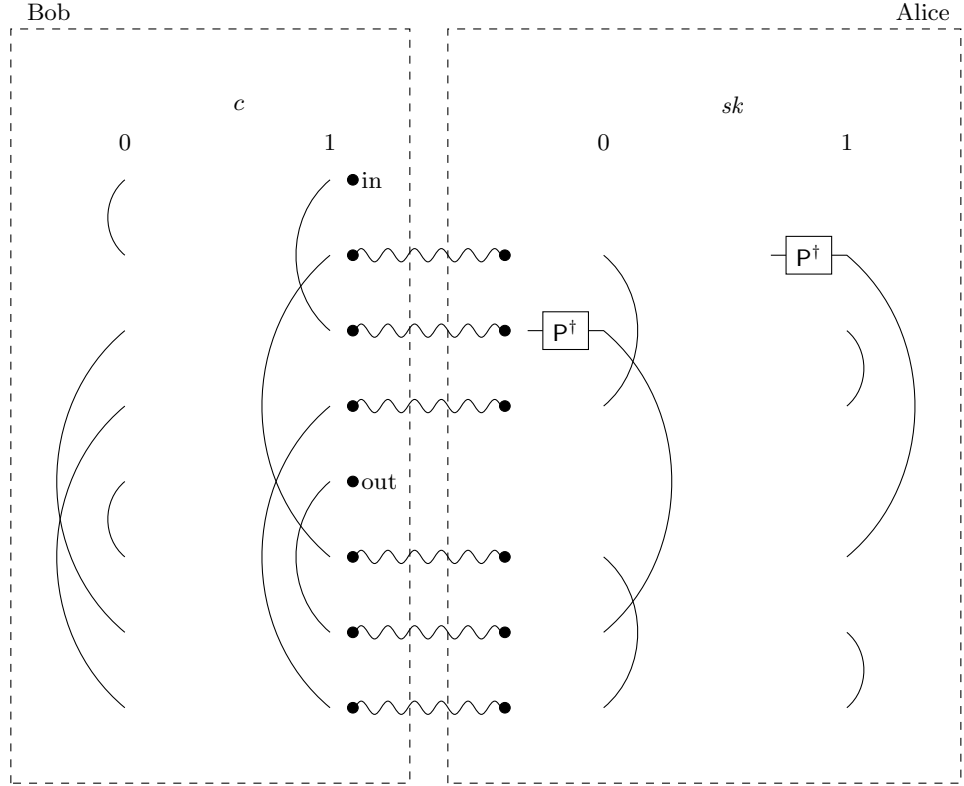
\begin{figure}[h!]
\centering
\begin{tikzpicture}[decoration=snake]
\filldraw (0,0) circle (2pt);
\filldraw (2,0) circle (2pt);
\draw[decorate] (0,0) -- (2,0);
\filldraw (0,1) circle (2pt);
\filldraw (2,1) circle (2pt);
\draw[decorate] (0,1) -- (2,1);
\filldraw (0,2) circle (2pt);
\filldraw (2,2) circle (2pt);
\draw[decorate] (0,2) -- (2,2);

\filldraw (0,3) circle (2pt);
\node[anchor=west] at (0,3) {in};

\node at (-1.5,4) {$c$};
\node at (-3,3.5) {$0$};
\node at (-0.3,3.5) {$1$};
\draw (-0.3,1) to [bend left=50] (-0.3,3);
\draw (-3,2) to [bend left=50] (-3,3);

\node at (5,4) {$\sk$};
\node at (3.3,3.5) {$0$};
\node at (6.5,3.5) {$1$};
\draw (3.3,0) to [bend right=50] (3.3,2);
\draw (6.5,0) to [bend right=50] (6.5,1);

\filldraw (0,-4) circle (2pt);
\filldraw (2,-4) circle (2pt);
\draw[decorate] (0,-4) -- (2,-4);
\filldraw (0,-3) circle (2pt);
\filldraw (2,-3) circle (2pt);
\draw[decorate] (0,-3) -- (2,-3);
\filldraw (0,-2) circle (2pt);
\filldraw (2,-2) circle (2pt);
\draw[decorate] (0,-2) -- (2,-2);

\filldraw (0,-1) circle (2pt);
\node[anchor=west] at (0,-1) {out};

\draw (-0.3,-3) to [bend left=50] (-0.3,-1);
\draw (-3,-2) to [bend left=50] (-3,-1);

\draw (3.3,-4) to [bend right=50] (3.3,-2);
\draw (6.5,-4) to [bend right=50] (6.5,-3);

\draw (-3,-3) to [bend left=50] (-3,1);
\draw (-3,-4) to [bend left=50] (-3,0);
\draw (-0.3,-2) to [bend left=50] (-0.3,2);
\draw (-0.3,-4) to [bend left=50] (-0.3,0);

\draw (3.3,-3) to [bend right=50] (3.3,1) to (2.3,1);
\draw (6.5,-2) to [bend right=50] (6.5,2) to (5.5,2);
\filldraw[fill=white] (5.7,1.75) rectangle (6.3,2.25);
\filldraw[fill=white] (2.5,0.75) rectangle (3.1, 1.25);
\node at (6,2) {$\P^\dag$};
\node at (2.8,1) {$\P^\dag$};

\draw[dashed] (-4.5,-5) rectangle (0.75,5);
\draw[dashed] (1.25,-5) rectangle (8,5);
\node[anchor=south] at (-4,5) {Bob};
\node[anchor=south] at (7.5,5) {Alice};

\end{tikzpicture}
\caption{Removal of a possible $\P$ error using two copies of the garden-hose protocol for $\TOY.\Dec$. For example, if $c = 0$ and $\sk = 1$, the qubit is teleported through EPR pairs 1 and 4, with a $\P^\dag$ applied to it by Alice in between. The input qubit always ends up on position `out'.}
\label{fig:gadget-construction-2}
\end{figure}

After the execution of the protocol, the potential $\P$ error on the qubit $\P^a\X^a\Z^b\ket\psi$ has  been removed, but additional Pauli transformations also have occurred as a result of the teleportations. The exact transformations depend on both the path the qubit has taken and the measurement outcomes of Alice and Bob. Therefore, Alice has to send all of this information (homomorphically encrypted) to Bob, so that he can update his keys to reflect the new state $\X^{a'}\Z^{b'}\ket\psi$ of the qubit.

Since the order of the Bell measurements does not influence the outcome of the protocol, Alice can perform her part of the protocol already during the key-generation phase. She starts by generating enough EPR-pairs for the gadget (six in this example), and performs the measurement on her own halves of the EPR-pairs. Effectively, this action generates six qubits that are entangled in some way that depends on $\sk$ (see Figure~\ref{fig:gadget-construction-3}). Because of the random Pauli's that arise from the Bell measurements, Bob is not able to tell which pairs are connected without knowing Alice's measurement outcomes. To him, the state of the gadget is completely mixed (see Equation~\ref{eq:mixedgadget}).

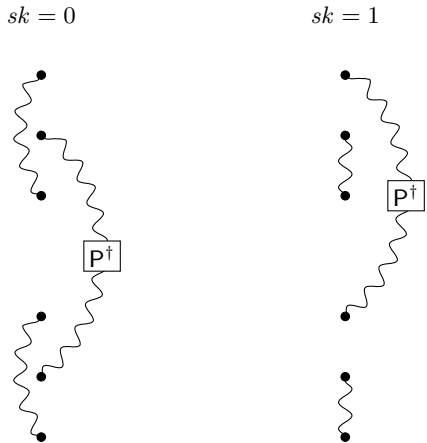
\begin{figure}[h!]
\centering
\begin{tikzpicture}[decoration=snake, scale=0.8]

\filldraw (5,0) circle (2pt);
\draw[decorate] (5,0) to (5,1);
\filldraw (5,1) circle (2pt);
\filldraw (5,2) circle (2pt);
\draw[decorate] (5,-2) to [bend right=50] (5,2);
\filldraw (5,-4) circle (2pt);
\draw[decorate] (5,-4) to (5,-3);
\filldraw (5,-3) circle (2pt);
\filldraw (5,-2) circle (2pt);

\filldraw (0,0) circle (2pt);
\draw[decorate] (0,0) to [bend left] (0,2);
\filldraw (0,1) circle (2pt);
\filldraw (0,2) circle (2pt);
\draw[decorate] (0,-3) to [bend right=50] (0,1);
\filldraw (0,-4) circle (2pt);
\draw[decorate] (0,-4) to [bend left] (0,-2);
\filldraw (0,-3) circle (2pt);
\filldraw (0,-2) circle (2pt);

\node at (0,3) {$\sk = 0$};
\node at (5,3) {$\sk = 1$};

\filldraw[fill=white] (0.7,-1.25) rectangle (1.3,-0.75);
\filldraw[fill=white] (5.7,-0.25) rectangle (6.3,0.25);
\node at (1,-1) {$\P^\dag$};
\node at (6,0) {$\P^\dag$};

\end{tikzpicture}
\caption{The two possible gadgets $\Gamma(\sk)$ that $\SCHEME.\GenGadget$ might generate for $\TOY.\Dec$. Effectively, a gadget consists of $2\gh(\TOY.\Dec)$ EPR-pairs, ordered in a way that depends on $\sk$. Some EPR-pairs have an additional $(\P^\dag \otimes I)$ transformation applied to them. The evaluator's input $c$ determines whether or not the input qubit is teleported through such a transformation, but an evaluator is unable to tell whether it is.}
\label{fig:gadget-construction-3}
\end{figure}

\section{Gadget construction for Learning With Errors}
\label{app:bv11construction}
\subsubsection*{Small modulus.}
Take the modulus $p$ to be polynomial in $\kappa$.
We describe a series of small `permutation gadgets' that move an arbitrary qubit to a location, depending
on whether $m=0$ or $m=1$. By doubling the construction as seen before, it is easy to turn these into a gadget
which applies an inverse phase gate whenever $m=1$. Note that we could just apply Theorem~\ref{thm:logspace} in order to construct a gadget directly from a log-space Turing machine of the decryption function. In this example, however, we choose to exhibit a more efficient gadget that exploits the structure of the BV11 scheme.

We follow \cite[Section 4.5]{BV11} in rewriting Equation~\ref{eq:bv11decrypt} in terms of binary arithmetic. Let
$\mathbf{s}[i](j)$ denote the $j$th bit of the $i$th entry of $\mathbf{s}$, then the inner product can be written as
\begin{align}
w - \langle \mathbf{v}, \mathbf{s}\rangle  \pmod p &= w - \sum^k_{i=1} \mathbf{v}[i] \mathbf{s}[i] \pmod p \nonumber \\
&= w - \sum^k_{i=1} \sum^{\log p}_{j=0}  \mathbf{v}[i](j) \cdot 2^j \cdot \mathbf{s}[i] \pmod p \label{eq:bv11decomp}
\end{align}

Let a \emph{permutation gadget} be a subgadget of size $2p$, parametrized
by a number $q \in \mathbb{Z}_p$. Label the first $p$ qubits by
$0_\mathrm{in}$ to $(p-1)_\mathrm{in}$, and the second $p$ qubits 
by $0_\mathrm{out}$ to $(p-1)_\mathrm{out}$. The gadget simply
creates EPR pairs between $x_\mathrm{in}$ and $({x+q \pmod p})_\mathrm{out}$,
for all $x \in \mathbb{Z}_p$.
Such a gadget can effectively simulate addition with $q$ over $\mathbb{Z}_p$.

For each element of the vector $\mathbf{s}$ we will create 
$\log p$ permutation gadgets. The intuition behind
the construction is as follows: The inner product which computed the decryption
of the ciphertext is written as a sum of $\kappa \log p$ numbers, that 
either contribute to the sum
or not, depending on a bit of the ciphertext $\mathbf{v}$.

For each $i$ from 1 to $\kappa$, and each $j$ from $0$ to $\log p$,
we create a permutation gadget, labeled by $(i,j)$, for the number $2^j \cdot \mathbf{s}[i]$.

The evaluator uses this gadget in the following way. He performs a Bell measurement between the input qubit and the $0_\mathrm{in}$ qubit
of the first gadget such that $\mathbf{v}[i](j) = 1$.
Then, he connects all output qubits  $0_\mathrm{out}$ to $(p-1)_\mathrm{out}$
of this gadget to all the input qubits
of the next gadget for which $\mathbf{v}[i](j) = 1$.

After teleporting his qubit through all gadgets, the qubit will be exactly 
at the location $z_\mathrm{out}$ of the final gadget the evaluator used,
where $z = \sum^\kappa_{i=1} \mathbf{v}[i] \mathbf{s}[i] \pmod p$. (Although
of course the evaluator has no way of knowing which of the $p$ locations this is.)
He can then, by simple permutation, apply an inverse phase gate whenever $w - z \pmod 2 = 1$.

Finally, as in the construction from the previous section,
we double the entire construction to route the unknown qubit back to a known location.
The size of the total gadget is then bounded by $4 \kappa p \log p$.

\subsubsection*{Large modulus.}
In case the modulus $p$ is superpolynomially large, constructing the gadget
explicitly appears to be much harder, and a log-space algorithm for this inner product is not immediately obvious.
For completeness, we sketch a proof strategy to reiterate that such a polynomial-sized gadget does still exist in this case.

The decryption function of Equation~\ref{eq:bv11decrypt}
has depth $O(\log \kappa + \log \log p)$,
see for example \cite[Lemma 4.5]{BV11}. This can be proven
by writing the decomposition of Equation~\ref{eq:bv11decomp} as a Wallace tree. 

Given a low-depth circuit, we could now apply Theorem~\ref{thm:gadget-logdepth}
to convert this circuit into a garden-hose protocol. In contrast to the small-modulus case, we do not exploit the structure of the decryption function to construct a more efficient gadget.

\section{Constructing gadgets using swap and Paulis}\label{app:gadget-construction-limited-quantum-power}
In the current description of the gadget generation, the key generator has to be able to perform a variety of tasks: he has to generate EPR-pairs, as well as perform $\P^{\dag}$ gates and Bell measurements. We show in this section how the gadgets can be generated securely using only $\X$, $\Z$ and $\cnot$, when the key generator is given resources by some computationally more powerful (but potentially malicious) party, for example the evaluator.

As described in Section~\ref{sec:gadget}, we see from Figure~\ref{fig:gadget-construction-3} that the gadget $\Gamma_{\pk'}(\sk)$ is effectively a list of $2m$ EPR-pairs (some of which have an extra $(\P^\dag \otimes I)$ transformation on them), with the qubits ordered in some way that depends on $\sk$. If the key generator is supplied with a list of $2m$ EPR-pairs $\ket{\Phi^+}$ and as many pairs $(I \otimes \P^\dag)\ket{\Phi^+}$, it is clear that he can create the gadget by swapping some of the qubits (using $\cnot$ gates), and applying random Pauli operations (using $\X$ and $\Z$ gates) on every pair. Any unused pairs are discarded.

If the supplier of these pairs follows the protocol and sends actual EPR-pairs to the key generator, this tactic suffices to hide all information about $\sk$. However, if the supplier acts maliciously, he may send two qubits to the key generator claiming that they form an EPR-pair, while in reality he is keeping some form of entanglement with one or both of the qubits. We need to make sure that even in this case, where the supplier actively tries to gather information about $\sk$, this information is still secure.

The key generator, upon receiving the (real or fake) EPR-pairs, can apply independently selected random Pauli transformations on every qubit. If the qubits really formed EPR-pairs, it would suffice to apply a random Pauli to only one of the two qubits in the pair, but by applying this transformation to both qubits, any entanglement that a malicious supplier might hold with any of them becomes completely useless. Since any swap of two qubits consists of three $\cnot$ gates that commute with the Pauli's, the state after swapping the qubits into the correct order is still completely mixed. Hence, no information about $\sk$ is revealed to the supplier.

\end{document}